\documentclass[reprint,
superscriptaddress,
 aps,
pra,10pt,longbibliography
]{revtex4-2}

\usepackage[dvipsnames]{xcolor}
\usepackage[normalem]{ulem}
\usepackage{graphicx}
\usepackage{dcolumn}
\usepackage{bm}

\usepackage{hyperref}
\hypersetup{colorlinks = true, linkcolor = magenta, urlcolor = blue, citecolor = magenta}
\usepackage{amsmath,amssymb}
\usepackage{tcolorbox}
\usepackage{setspace}
\usepackage{stmaryrd}
\usepackage{amssymb}

\newtheorem{lemma}{Lemma}
\newtheorem{proposition}{Proposition}

\usepackage{braket}
\usepackage{xcolor}
\usepackage{mathtools}

\newenvironment{proof}[1][Proof]{\noindent\textbf{#1.} }{\ \rule{0.5em}{0.5em}}

\graphicspath{{../images/}}
\usepackage[toc,page]{appendix}
\usepackage{mathtools}

\newcommand{\Tr}{\operatorname{Tr}}
\newcommand{\tr}{\operatorname{Tr}}

\newcommand{\curlN}{\mathcal{N}}
\newcommand{\curlM}{\mathcal{M}}

\usepackage{braket}

\def\>{\rangle}
\def\<{\langle}

\usepackage[bottom]{footmisc}

\newcommand{\mmw}[1]{{\color{purple} MMW: #1}}

\allowdisplaybreaks

\begin{document}

\title{Efficiently Computable Strategies and Limits for Bosonic Channel Discrimination}

\author{Zixin Huang}
\email{zixin.huang@rmit.edu.au}
\affiliation{School of Science, STEM College, RMIT University, Melbourne, VIC 3000, Australia}

\author{Ludovico Lami}
\affiliation{Scuola Normale Superiore, Piazza dei Cavalieri 7, 56126 Pisa, Italy}
\affiliation{Korteweg--de Vries Institute for Mathematics, University of Amsterdam, Science Park 105-107, 1098 XG Amsterdam, the Netherlands}

\author{Vishal Singh}
\affiliation{Mathematical Quantum Information RIKEN Hakubi Research Team, 
RIKEN Pioneering Research Institute (PRI) and RIKEN Center for Quantum Computing (RQC), Wako, Saitama 351-0198, Japan}

\author{Mark M. Wilde}\email{wilde@cornell.edu }
\affiliation{School of Electrical and Computer Engineering, Cornell University, Ithaca, New
York 14850, USA}

\begin{abstract}
Discriminating between noisy quantum processes is a central primitive for quantum communication, metrology, and computing. While discrimination limits for finite-dimensional channels are well understood, the continuous-variable setting—particularly under experimentally relevant energy constraints—remains significantly less developed. In this work, we establish an energy-constrained chain rule for the Belavkin–Staszewski channel divergence, which yields a fundamental upper bound on the error exponents achievable by fully adaptive, energy-constrained quantum channel discrimination protocols. 
We then derive efficiently computable bounds on asymmetric error exponents for energy-constrained discrimination of bosonic dephasing and loss-dephasing channels. Specifically, we show that three operationally relevant quantities---the measured relative entropy, the Umegaki relative entropy, and the geometric R\'enyi divergence---admit semidefinite program (SDPs) formulations when the input energy is bounded and the Hilbert space is suitably truncated. Applying these tools, we demonstrate that optimal probes for these channels under energy constraints are Fock-diagonal, and we also enable numerically precise evaluation of bounds on achievable error exponents across discrimination strategies ranging from separable to fully adaptive. The resulting SDPs provide practical benchmarks for quantum-limited sensing in low-energy bosonic platforms.
\end{abstract}

\maketitle


\section{Introduction}

The ability to distinguish between different physical processes is a central problem in quantum information theory, with implications for quantum communication, metrology, and computing~\cite{bae2015quantum,bae2015quantum,bergou2010discrimination,sacchi2005,Piani2009,harrow2010adaptive}. Quantum channel discrimination generalizes quantum state discrimination to dynamical maps---completely positive, trace-preserving maps---that describe how quantum states evolve under noise, control, or measurement. The goal is to optimally determine which of several candidate channels acted on a probe system.

Dephasing and loss channels play a fundamental role in quantum information theory, as they model the most common noise mechanisms in quantum systems. Loss is the dominant impairment in optical quantum technologies, whereas dephasing governs decoherence in systems such as Bose--Einstein condensates~\cite{khodorkovsky2009decoherence,liu2010quantum}, trapped ions~\cite{valahu2024quantum}, and cold neutral atoms~\cite{carnio2015robust,dorner2012quantum,gross2010nonlinear}. Classes of AC signals can also be mapped to such channels~\cite{dey2024quantum}. In bosonic platforms, dephasing channels describe phase noise in optical and microwave modes, while loss channels capture photon attenuation; both are crucial for realistic modeling of quantum communication and metrology setups. Understanding the ultimate limits of discriminating among these channels is therefore of both theoretical and practical interest, with direct applications to quantum sensing~\cite{PhysRevLett.125.180502,shapiro2020quantum}, communication~\cite{cooney2016strong}, and computation~\cite{PhysRevA.106.032409,PhysRevA.105.032401}.

\begin{figure}[h!t]
\includegraphics[width=\linewidth]{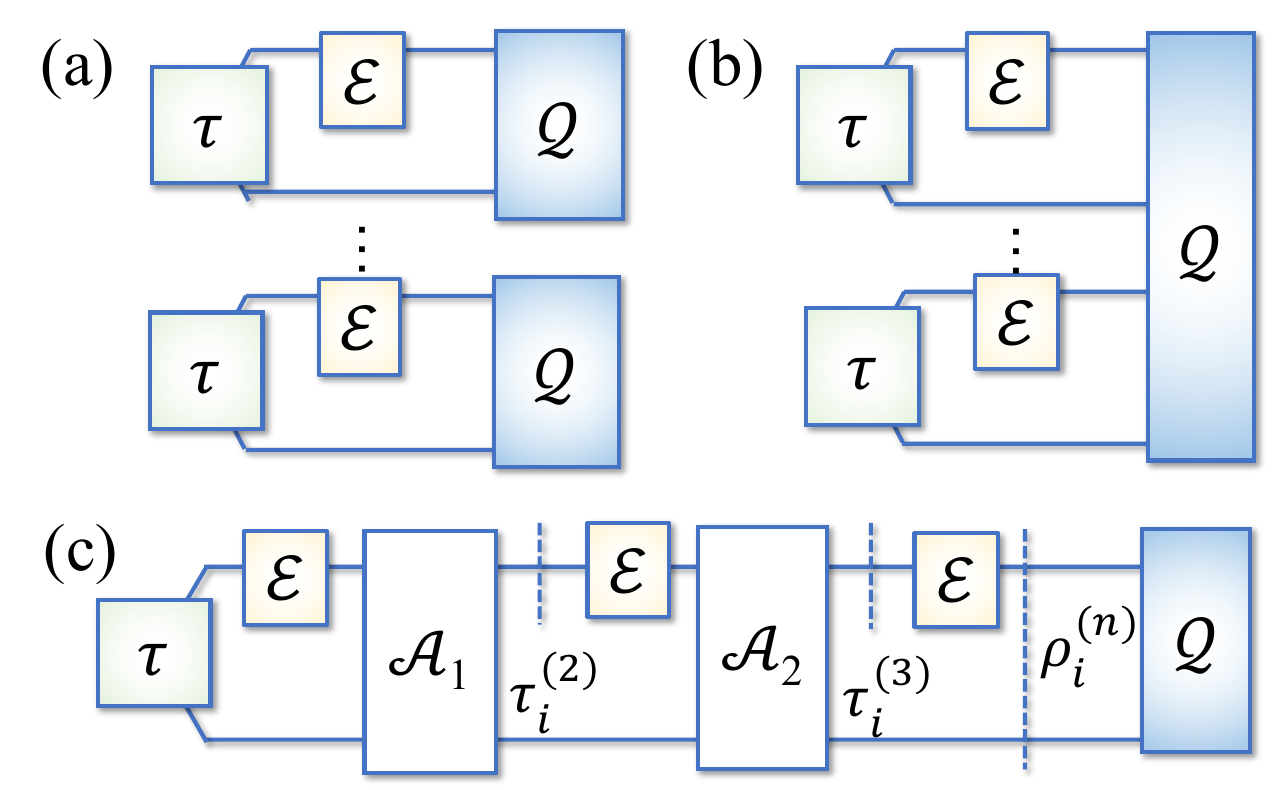}
\caption{\label{fig:scheme} 
Tiered quantum channel discrimination strategies.
The channel $\mathcal{E}$ is applied, chosen from $(\curlN,\curlM )$: (a) the measured relative entropy, requiring only product measurements of the form depicted; (b) the quantum relative entropy, in general requiring a collective measurement;
(c) a general, adaptive protocol for channel discrimination, when the channel is called three times. The initial input state is $\tau$, the adaptive operations are 
$\mathcal{A}_1$ and $\mathcal{A}_2$, and the final measurement is~$\mathcal{Q}$. 
The final states are denoted by $\rho_0^{(n)}$ and $\rho_1^{(n)}$, and $n=3$ in this case.}
\end{figure}

Quantum channel discrimination has been studied in the non-asymptotic regime, where one optimizes probe states, measurements, and possible ancilla assistance, and even for the most general adaptive protocols with feedforward~\cite{acin2001statistical,sacchi2005minimum,hayashi2009discrimination,wilde2020amortized,Bergh2024,bergh2023infinite}. In the asymptotic i.i.d.\ setting, performance is quantified by error exponents. In the symmetric regime, the optimal exponent is governed by channel Chernoff--type quantities, whose characterization can depend on whether adaptive strategies are allowed~\cite{hayashi2009discrimination,wilde2020amortized,KatariyaWilde2021,salek2022,ji2025bary}. In the asymmetric (Stein) case, channel versions of the quantum Stein's lemma identify the relevant metric as the channel relative entropy or its regularization~\cite{cooney2016strong,wang2019,wilde2020amortized,fang2020,Bergh2024,bergh2023infinite,bergh2025composite}. 

In continuous-variable systems, discrimination of bosonic Gaussian channels has been analyzed previously~\cite{takeoka2008discrimination,nair2011discriminating,takeoka2016,pirandola2017,nair2022,nair2023}. However, the corresponding limits for non-Gaussian and energy-constrained channels remain largely uncharacterized. Recent work established exact quantum limits for bosonic dephasing channels without imposing energy constraints, showing that the problem reduces to classical statistics over phase-noise distributions~\cite{huang2024exact}. 
In practice, the probe energy is finite---either due to constraints on the mean photon number or to prevent nonlinear effects that distort the channel. 
Incorporating such energy constraints into the channel-discrimination framework is therefore essential.
We note that energy-constrained discrimination of bosonic dephasing channels has been considered recently in the symmetric setting~\cite{oskouei2024discrimination}.

In this work, we compare quantum channel discrimination strategies for different levels of complexity in probe preparation and measurement (Fig.~\ref{fig:scheme}). Here we focus exclusively on the asymmetric setting of channel discrimination with energy constraints. The goal is to distinguish one quantum channel $\curlN$ from another, $\curlM$, by querying the channel $n$ times. 
We analyze the performance of these strategies and show that corresponding exponential error rates can be formulated as semidefinite programs.
This formulation allows for direct computation of minimum-error exponents while naturally incorporating energy constraints on input states. Our results are also informative for Gaussian settings, where energy constraints play a similar role in limiting the distinguishability.

The novelty and main contributions of this work are: first, we establish an energy-constrained chain rule for the Belavkin–Staszewski channel divergence (Eq.~\eqref{eq:chain-rule0}), which yields an upper bound on the performance of energy-constrained, fully adaptive discrimination protocols. Second, we systematically relate three levels of discrimination complexity--fully adaptive strategies, collective measurements, and separable measurements--to the Belavkin-Staszewski divergence, quantum relative entropy, and the measured relative entropy, respectively (Sec.~\ref{sec:q-rel-ents}); we provide an SDP formulation or approximation for the quantity associated with each level. Third, for bosonic dephasing and loss-dephasing channels, phase covariance reduces the probe optimization to Fock-diagonal states, making the bounds numerically accessible (Sec.~\ref{sec:applications}). Physically, these results provide computable benchmarks for assessing how much is gained by moving from simple measurements to collective or adaptive strategies under a fixed photon-number budget, and they identify low-energy probe states that can substantially outperform standard coherent-state strategies.

Our approach builds on a sequence of results linking quantum divergences, convex optimization, and channel discrimination. Belavkin and Staszewski~\cite{Belavkin1982} introduced a quantum generalization of relative entropy different from the standard one defined by Umegaki. Matsumoto~\cite{Matsumoto2018} defined the maximal quantum $f$-divergence and provided variational forms well-suited to semidefinite programming. Fang and Fawzi~\cite{FangFawzi2021} extended these constructions to quantum channels via the geometric R\'enyi divergence, yielding computable single-letter bounds. These developments motivate our SDP-based treatment of energy-constrained bosonic channel discrimination, an open setting where such methods yield physically relevant benchmarks.

The structure of this paper is as follows. Section~\ref{sec:q-ch-disc} reviews the framework of quantum channel discrimination, with an emphasis on asymmetric hypothesis testing, relevant divergences (Section~\ref{sec:q-rel-ents}), and bounds on the error exponent (Section~\ref{sec:err-exp-bnds}). Section~\ref{sec:SDPs-channel-rel-ents} reviews how to formulate various channel relative entropies as semidefinite programs. Section~\ref{sec:results} presents our main results: a hierarchy of error exponents corresponding to optimal measurements under energy constraints, followed by numerical evaluations and an analysis of optimal probes. We provide an outlook  in Section~\ref{sec:conclusion}. The appendices provide detailed derivations.

\section{Quantum channel discrimination} 

\label{sec:q-ch-disc}

The task of quantum channel discrimination is to determine whether an unknown channel is $\curlN$ or $\curlM$, given access to $n$ queries and $n-1$ adaptive operations in between (Fig.~\ref{fig:scheme} c). 
%
In the asymmetric setting, the objective is to minimize the type~II error probability $\beta_n$ under a constraint on the type~I error probability. 
Let $\mathcal{Q} := (Q_0,Q_1)$ denote the quantum measurement performed at the end of the protocol, 
and let $\rho_0^{(n)}$ and $\rho_1^{(n)}$ denote the states at the output of the adaptive protocol. 
We define the type-I error probability $\alpha_n$ as the probability of deciding $\curlM$ when $\curlN$ is queried 
and the type-II error probability $\beta_n$ as the probability of deciding $\curlN$ when $\curlM$  is queried:
\color{black}
The error probabilities are given by  
\begin{align}
\alpha_n \coloneqq  \tr[Q_1 \rho_0^{(n)}], \qquad \beta_n \coloneqq \tr[Q_0 \rho_1^{(n)}].
\end{align}

To impose an energy constraint, let $H\geq 0 $ be a Hamiltonian acting on the probe system $A$,
and $E\geq 0$ be the allowed average input energy. For the bosonic channels considered below, we take
\begin{align}
H = \hat n = \sum_{m = 0}^\infty m |m\rangle\!\langle m| . 
\end{align}
\color{black}
In energy-constrained channel discrimination, a Hamiltonian $H$ (taken to be positive semidefinite for convenience) is specified, as well as an energy constraint $E \geq 0$. The following additional average energy constraint applies for $i \in \{\mathcal{N},\mathcal{M}\}$:
\begin{equation}
    \Tr[H \overline{\tau}_{i}]\leq E, \label{eq:energy-constraint}
\end{equation}
where $\overline{\tau}_{i} \coloneqq \frac{1}{n}\sum_{j=1}^n \tau^{(j)}_i $, with $\tau^{(j)}_i$ being the reduced state at the input of the $j$th channel use (see Fig.~\ref{fig:scheme} c). Let us define the optimal type II error probability under a fixed energy constraint $E$ and type I error probability constraint as follows:
\begin{equation}
    \beta^\star_n(E,\varepsilon) \coloneqq \min_{\mathcal{A}_E^{(n)}} \{ \beta_n : \alpha_n \leq \varepsilon\},
\end{equation}
where the minimization is over every channel discrimination strategy $\mathcal{A}^{(n)}_E$ of the aforementioned form that obeys the energy constraint in~\eqref{eq:energy-constraint}.

\subsection{Quantum relative entropies}

\label{sec:q-rel-ents}

There are several quantities of interest that are helpful in establishing bounds in quantum hypothesis testing. 
The most widely used is the quantum relative entropy~\cite{Ume62}:
\begin{align}
D(\rho \| \sigma) \coloneqq  
&\begin{cases}
\mathrm{Tr}\!\left[ \rho \, (\ln  \rho - \ln  \sigma) \right], & \text{if } \mathrm{supp}(\rho) \subseteq \mathrm{supp}(\sigma), \\
+\infty, & \text{otherwise}.
\end{cases}
\end{align}
Another family of quantum relative entropies consists of 
the geometric R\'enyi divergence (GRD)~\cite{matsumoto2015new,Matsumoto2018}; for $\alpha \in (1,2]$,
\begin{align} \label{eq: state_grd}
\widehat D_\alpha(\rho\|\sigma) & \coloneqq \frac{1}{\alpha-1} \ln  \tr 
\!\left[\sigma\left(\sigma^{-\frac{1}{2}} \rho \sigma^{-\frac{1}{2}}\right)^\alpha\right].
\end{align}
The GRD is also known as the maximal R\'enyi divergence because it is the maximal quantum generalization of the R\'enyi divergence that satisfies additivity and the data-processing inequality~\cite{matsumoto2015new,Matsumoto2018}. 
The GRD is monotonically increasing with $\alpha$ \cite[Prop.~72]{KatariyaWilde2021}. For $\alpha\rightarrow 1$, the GRD converges to the Belavkin--Staszewski relative entropy~\cite[Prop.~79]{KatariyaWilde2021}, defined as~\cite{Belavkin1982}
\begin{align}
\widehat D(\rho \|\sigma) \coloneqq  \tr[\rho \ln (\rho^{1/2}\sigma^{-1}\rho^{1/2})],
\end{align}
 which in turn upper bounds the Umegaki relative entropy~\cite{HiaiPetz1993}.
Another is  the measured relative entropy~\cite{Donald1986,Piani2009measured}:
\begin{align}
\label{eq:mre_states}
D^M(\rho \|\sigma) \coloneqq  \sup_{\left(\Lambda_x\right)_x} \sum_{x} \tr[\Lambda_x\rho] 
\ln\!\left(\frac{\tr[\Lambda_x\rho]}{\tr[\Lambda_x\sigma]}\right),
\end{align}
where each $\Lambda_x$ is a rank-one positive semidefinite operator such that $ \sum_{x} \Lambda_x =\openone$. See App.~\ref{app:q-rel-ents} for more details of quantum relative entropies.

\subsection{Bounds for error exponents of quantum channel discrimination}

\label{sec:err-exp-bnds}

Each of these quantities admits a corresponding formulation for quantum channels with energy-constrained input states, and they either
directly correspond to or can be related to achievable error exponents corresponding to different strategies. Specifically (Fig.~\ref{fig:scheme}), these are 
\begin{enumerate}
\item[(a)] Non-collective measurements: characterized by the energy-constrained measured relative entropy $D^M_{H,E} (\curlN\|\curlM)$,
\item[(b)] Collective measurements: characterized by the energy-constrained relative entropy $D_{H,E}(\curlN\|\curlM)$,
\item[(c)] Adaptive strategies: upper bounded by the energy-constrained Belavkin--Staszewski divergence $\widehat D_{H,E}(\curlN\|\curlM)$ or the energy-constrained GRD $\widehat D_{\alpha, H,E}(\curlN\|\curlM)$ for $\alpha \in (1,2]$, as proven in Section~\ref{sec:BS-GRD-bound}. 
\end{enumerate}

For quantum channels $\curlN$ and $\curlM$, the energy-constrained measured relative entropy, relative entropy, Belavkin--Staszewski relative entropy, and GRD are defined respectively for $E\geq 0$ and a Hamiltonian $H_A$ acting on the channel input system $A$ as~\cite[Eq.~(12.12)]{sharma2018bounding}
\begin{multline}
\label{eq:channel_gen_div}
  \mathbb{D}_{H,E}(\curlN \| \curlM) \coloneqq \\
\sup_{\rho_{RA}:  \tr[H_A\rho_A] \leq E}   \mathbb{D} \!\left(\curlN_{A\rightarrow B}(\rho_{RA}) \| \curlM_{A\rightarrow B}(\rho_{RA})\right), 
\end{multline}
where $\mathbb{D}$ can be $D^M$, $D$, $\widehat{D}$, or $\widehat{D}_\alpha $.
The reference system~$R$ allows us to consider the most general possible input state allowed by quantum mechanics. These divergences each admit variational characterizations. Under an energy constraint and a Fock-space cutoff,   these variational forms reduce to finite-dimensional semidefinite programs (SDPs). 

The following inequalities hold for the above quantities:
\begin{equation}
\begin{aligned}
\label{eq:inequalities}
D^M_{H,E} (\curlN\|\curlM)& \leq  D_{H,E}(\curlN\|\curlM) \\
& 
\leq \lim_{\varepsilon\to 0} \lim_{n\to \infty}\frac{-\ln \beta^\star_n(E,\varepsilon)}{n} .
\end{aligned}    
\end{equation}

\subsection{Semidefinite programs for channel relative entropies}

\label{sec:SDPs-channel-rel-ents}

If we restrict the energy of the input state $\rho_A$ such that $\tr[H \rho_A]\leq E$, the energy-constrained GRD can be written for $\alpha \in (1,2]$ as (see App.~\ref{sec:proof1_grd})
\begin{multline}
\label{eq:EC_geo_ch}
\widehat D_{\alpha,H,E}  (\curlN\|\curlM) = \\
\sup_{\substack{\rho_R:\tr[H\rho_R] \leq E}} \left\{\frac{1}{\alpha -1}  
                    \ln   \tr[\rho_R \tr_B[G_{1-\alpha}(J_{RB}^{\mathcal{N}},J_{RB}^{\mathcal{M}})]] \right\}.
\end{multline}
The reference system $R$ is isomorphic to $A$ and the following identity holds:
$\rho_R = \rho_A^T$. Consequently,
$\tr[H\rho_A]=\tr[H\rho_R^T] =\tr[H\rho_R]$.
Here $G_t$ is the weighted operator geometric mean, defined by~\cite{pusz1975functional,Kubo1980}. 
\color{black}
\begin{equation}\label{eq:geo_mean_defn}
    G_t(X,Y) \coloneqq X^{\frac{1}{2}}\!\left(X^{-\frac{1}{2}}YX^{-\frac{1}{2}}\right)^tX^{\frac{1}{2}}, \quad t\in \mathbb{R},
\end{equation}
and $J_{RB}^{\curlN}$ and $J_{RB}^{\curlM}$ are the Choi matrices of the two channels (see, e.g., \cite[Def.~4.1]{khatri2020principles} and~\eqref{eq:Choi_op_defn}). Here $G_{1-\alpha}$ can be computed directly, or, after a Hilbert space truncation, via an SDP approximation indexed by an integer~$\ell$, where we take $\alpha(\ell) \coloneqq 1 + 2^{-\ell}$. See App.~\ref{app:weighted-op-geo-mean} for more details of the weighted operator geometric mean.

After a Hilbert-space truncation, the energy-constrained channel relative entropy in~\eqref{eq:channel_gen_div} (with $\mathbb{D} = D$ therein) can be approximated from below via the following SDP~\cite[Theorem~III.1]{kossmann2024}:
\begin{align} 
\label{sdp:channel_re} 
D_L \coloneqq 
\sup_{\substack{\rho_R\geq 0, \\ Q_k \geq 0}} 
\left\{
\begin{array} [c]{c}
 \sum_{k = 1}^r \tr \left[Q_k (\alpha_k J^\curlN_{RB} +\beta_k J^\curlM_{RB}) \right] \\
    + \ln \lambda + 1 - \lambda \qquad : \\
      Q_k \leq \rho_R \otimes \openone_B, ~\tr[\rho_R H] \leq E, \\ 
    \tr[\rho_R] =1
    \end{array}
\right\}.
\end{align} 

The approximation does not come from truncating the Hilbert space, but rather from the discretized representation of the quantum relative entropy. For a fixed finite-dimensional truncation, the channel relative entropy is expressed using an exact integral representation, and the integration interval is then discretized into $r$ subintervals. 
The number of matrices to optimise, $r$, is related to the error $\epsilon$ via $r = O(\sqrt{\lambda/\varepsilon )}$.  \color{black}
The parameters $\alpha_k$ and $\beta_k$ can be chosen according to~\cite[Eq.~(6)]{kossmann2025}, 
and $\lambda$ is related to the max-relative entropy of channels, which can be computed as the SDP
$\inf \{\lambda : J^\curlN_{RB} \leq \lambda J^\curlM_{RB}\}$.

An upper bound approximation of the energy-constrained relative entropy of channels has also been derived in~\cite[Theorem~III.2]{kossmann2024}.  For $k\in \{1,\ldots,r\}$, $\gamma_0 = 1$, $\delta_0 = -1$,
\begin{align}
\label{eq:channel_upper}
D_U \coloneqq 
\inf_{\substack{ x\in\mathbb{R},~y\geq 0, \\ N_0 \in \text{Herm}}} 
\left\{
\begin{array}{c}
x+y E + \ln \lambda+1-\lambda: \\
N_k \geq \gamma_k J^\curlN_{RB} + \delta_k J^\curlM_{RB},\quad
N_k \geq 0 \\
x \openone_R + y H_R \geq \tr_B \left[ \sum_{k=0}^r N_k \right]\\
\end{array}
\right\}
\end{align}
with coefficients $\gamma_k$ and $\delta_k$ chosen as in~\cite[Corollary 1, Appendix E]{kossmann2025}.
See App.~\ref{app:q-rel-ent-channels} for more details of~\eqref{sdp:channel_re} and~\eqref{eq:channel_upper}.

The measured relative entropy for states admits a convenient variational formulation involving the operator logarithm~\cite[Lem.~1 \& Thm.~2]{berta2017variational}. By using a semidefinite approximation of the matrix logarithm~\cite{fawzi2019semidefinite}, Ref.~\cite{huang2024semi}  showed that
the measured channel relative entropy  can be approximated via the following SDP, with error $O(e^{-(m+k)})$:
\begin{multline}
\label{sdp:measured_sdp}
D_{H,E}^{M}(\mathcal{N}\Vert\mathcal{M}) \approx \\
\sup_{\substack{\Omega,~{\rho_R>0}, \\ \Theta\in \operatorname{Herm}}}\left\{
\begin{array}
[c]{c}
\operatorname{Tr}[\Theta J^\curlN_{RB}]-\operatorname{Tr}[\Omega J^\curlM_{RB}]+1:\\
\operatorname{Tr}[\rho_R]=1,\quad\operatorname{Tr}[H\rho_R]\leq E,\\
T_1,\ldots,T_m,~Z_0,\ldots,Z_k \in \operatorname{Herm},\\
 X = \rho_R \otimes \openone, \qquad Z_0 = \Omega \\
\left \{ \left[
\begin{array}{cc}
 Z_i & Z_{i +1} \\
 Z_{i +1}  & X \\
\end{array}
\right] \geq 0 \right\} _{i = 0}^{k-1}, \\ \\
\sum_{j=1}^m w_j T_j = 2^{-k} \Theta, \\
\left \{ \left[
\begin{array}{cc}
 Z_k-X-T_j & -\sqrt{t_j} T_j \\
 -\sqrt{t_j} T_j  & X - t_j T_j\\
\end{array}
\right] \geq 0 \right\} _{j = 1}^{m}
\end{array}
\right\}  .
\end{multline}
Here $w_j$ and $t_j$ are the weights and nodes for the $m$-point Gauss--Legendre quadrature on the interval $[0,1]$. In App.~\ref{sec:measured_re}, we provide a more detailed proof of~\eqref{sdp:measured_sdp} compared to that given in~\cite{huang2024semi}.

\section{Results}

\label{sec:results}

Our main theoretical result is an energy-constrained chain rule for the Belavkin–Staszewski channel divergence. We then use this to prove the following upper bound on the optimal error exponent for energy-constrained channel discrimination (see App.~\ref{app:up-bnd-ecbs-proof}):
\begin{equation}
\label{eq:type_II_err_ub_bel_stas}
    \lim_{\varepsilon\to 0} \lim_{n\to \infty}\frac{-\ln \beta^\star_n(E,\varepsilon)}{n}
\leq \widehat D_{H,E}(\curlN\|\curlM).
\end{equation}
We also organize the relevant divergences into a hierarchy and compare their performance: Belavkin--Staszewski upper bounds adaptive strategies (known not to be achievable in general), the channel relative entropy characterizes performance achievable with collective measurements, and measured variants correspond to non-collective (non-adaptive) measurement strategies.

\subsection{Belavkin--Staszewski and geometric R\'enyi relative entropy}

\label{sec:BS-GRD-bound}

We start with the Belavkin--Staszewski divergence, because it yields the strongest upper bound on the achievable error exponent.
As mentioned above, a key theoretical contribution of our paper, underlying our results, is an energy-constrained chain rule for the Belavkin--Staszewski channel divergence (for proof, see Sec.~\ref{eq:sec_ECBS}):
\begin{multline}
\widehat{D} \!\left( \mathcal{N}_{A\to B}(\rho_{RA}) \middle\| \mathcal{M}_{A\to B}(\sigma_{RA}) \right) \\
\leq \widehat{D}(\rho_{RA} \| \sigma_{RA}) + \widehat{D}_{H,E}(\mathcal{N}\|\mathcal{M}),
\label{eq:chain-rule0}
\end{multline}
which holds for every state $\rho_{RA}$ satisfying the energy constraint (i.e., $\operatorname{Tr}[H_A\rho_A]\leq E$). We also prove an $n$-copy generalization of~\eqref{eq:chain-rule0}, as detailed in App.~\ref{app:up-bnd-ecbs-proof}. It plays a central role in bounding the error exponent for adaptive discrimination strategies. In an adaptive protocol, the channel inputs generally depend on previous measurement outcomes, leading to hypothesis-dependent probe states across channel uses. This result shows that, despite this dependence, the divergence accumulated on average in each channel use is bounded by a single-use energy-constrained channel divergence. As a consequence, even the most general fully adaptive discrimination protocols obey the same asymptotic upper bound.

We do not prove that the energy-constrained chain rule applies to the GRD. However, the bound $\widehat{D}\leq \widehat{D}_\alpha$ holds for all $\alpha > 1$, so that we can use the GRD as an upper bound on the Belavkin--Staszewski divergence. 
Indeed, the Belavkin--Staszewski divergence can be approximated from above with increasing accuracy by taking the parameter $\alpha$ of the GRD arbitrarily close to one; this captures the ``best possible'' discrimination rate compatible with energy constraints.

\begin{figure}[t]
  \centering
  \includegraphics[trim= {0 0 0 0cm},clip, width=1.0\linewidth]{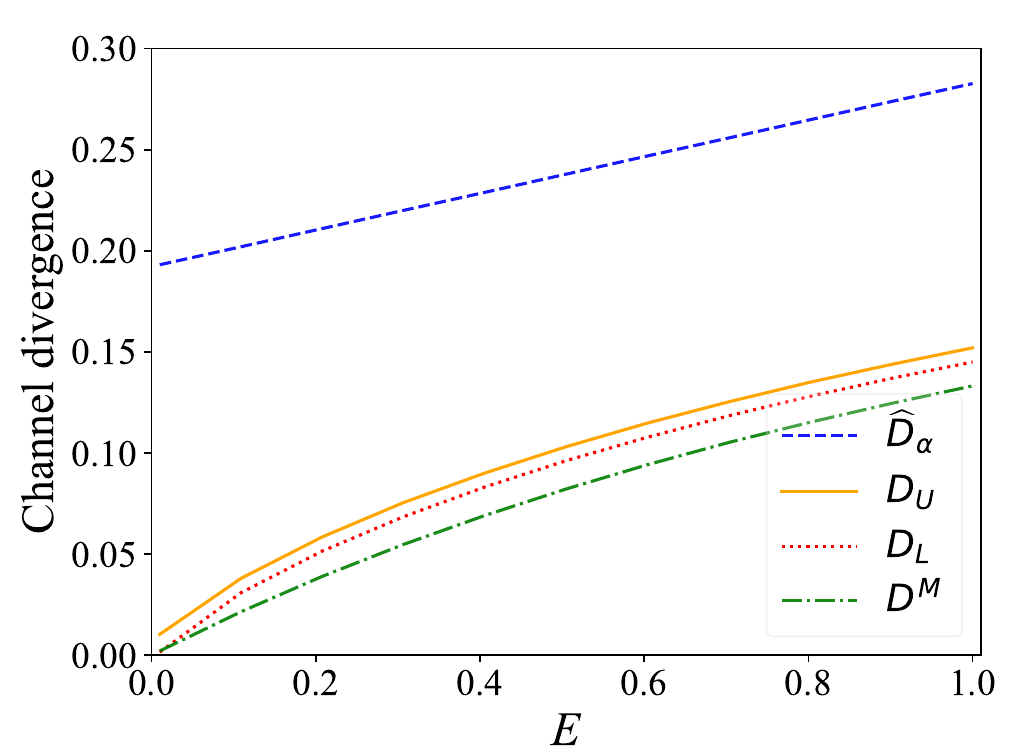}
  \caption{\label{fig:dephasingase_function_E} Channel divergences for a pure dephasing channel as a function of the energy constraint $E$; the parameters are $\gamma_1 = 0.1$ and $\gamma_2 = 0.4$. The dimension of the reduced Hilbert space of the probe satisfies $\dim(R)=9$. }
\end{figure}

\begin{figure}[t]
  \centering
  \includegraphics[trim= {0 0 0 00cm},clip, width=1.0\linewidth]{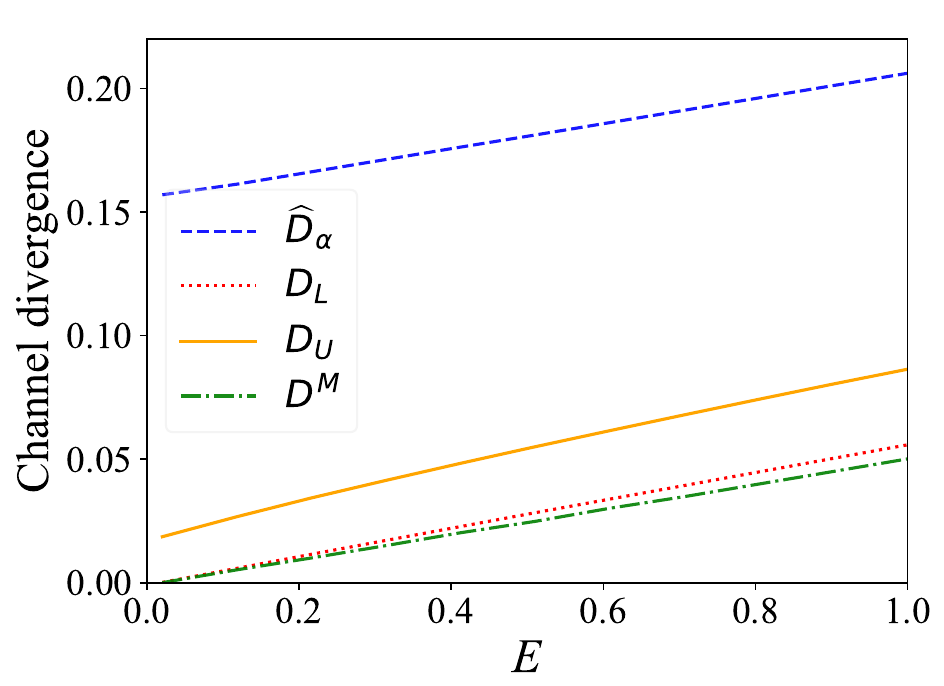}
  \caption{ \label{fig:loss_dephase_function_E}  Relative entropies for a loss-dephasing channel as a function of the energy constraint $E$. In this plot, the transmissivity parameters are $\eta_1 = 0.95$ and $\eta_2 = 0.85$, with a small dephasing parameter $\gamma_1=\gamma_2 = 0.01$. The dimension of the reduced Hilbert space of the probe satisfies $\dim(R)=9$. }
\end{figure}

\begin{figure}[t]
  \centering
  \includegraphics[trim= {0 0 0 00cm},clip, width=1.0\linewidth]{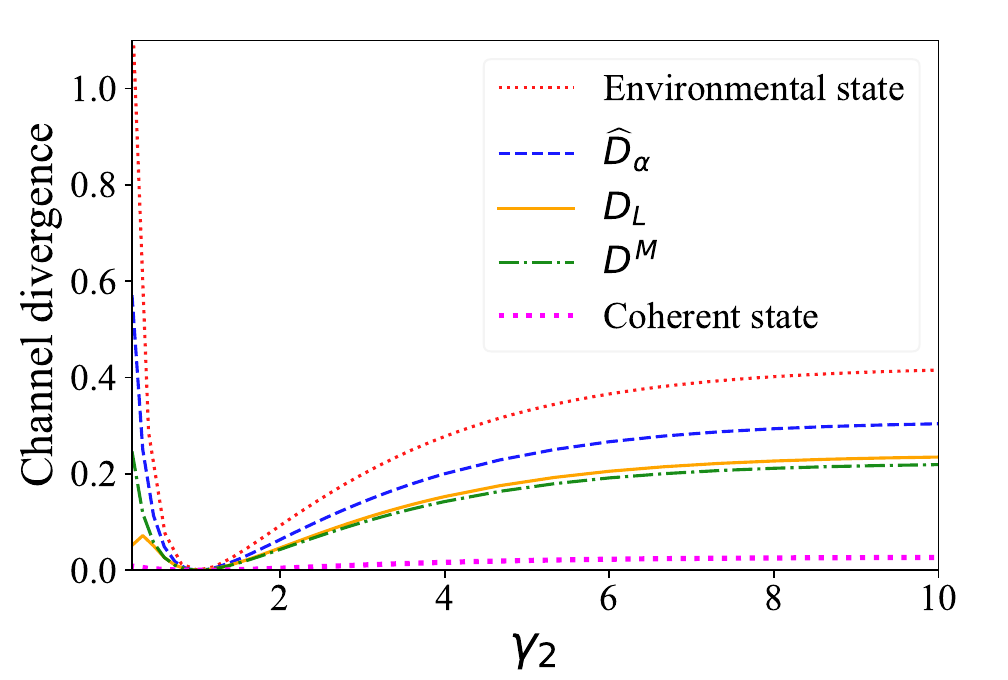}
  \caption{ \label{fig:function_gamma}
  Relative entropies for the bosonic dephasing channel as a function of $\gamma_2$. Here we fix $\gamma_1 = 1$ and the energy constraint $E=0.5$. The environmental state upper bounds the ultimate quantum limit without energy constraints (see Eq.~(31) of Ref.~\cite{huang2024exact}. }
\end{figure}

\subsection{Application to bosonic channels}
\label{sec:applications}

We apply the SDPs to the bosonic dephasing channel (BDC) and the loss-dephasing channel. 
The action of a BDC on a state $\rho$ is characterized by the probability density $p(\phi)$ for $\phi \in [-\pi,\pi]$ and can be written as follows:
\begin{align}
    \mathcal{D}_p(\rho) &\coloneqq  \int_{-\pi}^{\pi}p(\phi)e^{-i\hat{n}\phi}\rho e^{i\hat{n}\phi}, \\
     p_{\gamma}(\phi) &\coloneqq \frac{1}{\sqrt{2\pi\gamma}}\sum_{k=-\infty}^{\infty} e^{-\frac{1}{2\gamma}(\phi + 2\pi k)^2},
\end{align}
where $\hat{n}$ is the photon number operator. In the above, we have set the underlying probability density to be the wrapped normal distribution. The action of a pure-loss channel, with transmissivity $\eta$, on an arbitrary input state $\rho_A$ is  as follows:
\begin{equation}\label{eq:loss_ch_defn}
    \mathcal{L}_{\eta}\!\left(\rho_A\right) \coloneqq \operatorname{Tr}_E\!\left[U_{\operatorname{BS}}(\eta)\!\left(\rho_A\otimes |0\rangle\!\langle 0|_E\right)U_{\operatorname{BS}}(\eta)^{\dagger}\right].
\end{equation}
where $U_{\operatorname{BS}}(\eta)\hat{a}^{\dagger} U_{\operatorname{BS}}(\eta)^{\dagger} = \sqrt{\eta}~\hat{a}^{\dagger} + \sqrt{1-\eta}~\hat{b}^{\dagger}$. The bosonic loss-dephasing channel is defined as the serial concatenation of these two channels, i.e., $\mathcal{L}_\eta \circ \mathcal{D}_p = \mathcal{D}_p \circ \mathcal{L}_\eta$, where we note that these channels commute with each other. This channel has been previously studied in the context of communication and discrimination \cite{lami2023exact,leviant2022quantum,huang2024exact,annamele2024}.
We provide more detail of these channels, including their Choi operators, in App.~\ref{sec:CHOI}.

The SDPs for energy-constrained channel divergences can be simplified for bosonic dephasing and pure-loss channels. Indeed, we use a symmetry reduction specific to the bosonic channels considered here: since the channels are phase-insensitive, an energy-constrained optimal input can be taken to be Fock-diagonal, which simplifies the numerical optimization. The reasoning follows from Ref.~\cite{sharma2018bounding} and the discussion around Eqs.~(12.33)–(12.38) therein.

In Figs.~\ref{fig:dephasingase_function_E}-\ref{fig:loss_dephase_function_E}, we 
plot Eqs.~\eqref{eq:EC_geo_ch},~\eqref{sdp:channel_re},~\eqref{eq:channel_upper}, and~\eqref{sdp:measured_sdp}
as a function of the energy constraint $E$, for the dephasing and loss-dephasing channel, respectively. 
 We truncate the Hilbert space at Fock state equal to 8 ($\dim (R)  =9$); for the measured relative entropy, we take $m=k=3$, for the relative entropy we use $r = 13$, and for the GRD we take $\ell = 8$.
As expected, our numerical results  show that discrimination performance improves monotonically with the available energy.

Fig.~\ref{fig:function_gamma} provides a clear illustration of how energy-constrained probe optimisation fundamentally alters the achievable discrimination rates. We plot the relative entropy as a function of the dephasing parameter.  In addition, we compare a coherent state with the same mean photon number $|\alpha|^2 = 0.5$ (pink dotted line) with the ultimate limit without energy constraints, dictated by the bounds proven in Ref.~\cite{huang2024exact}. (Here, by ``coherent state,'' we mean a strategy involving a coherent-state input along with heterodyne detection at the output, as considered in~\cite{huang2024exact}.)
Even at a modest energy constraint on the input of 0.5 photons, the measured relative entropy significantly outperforms the coherent-state strategy.

\begin{figure}
  \centering
  \includegraphics[trim= {0 0 0 0cm},clip, width=1.0\linewidth]{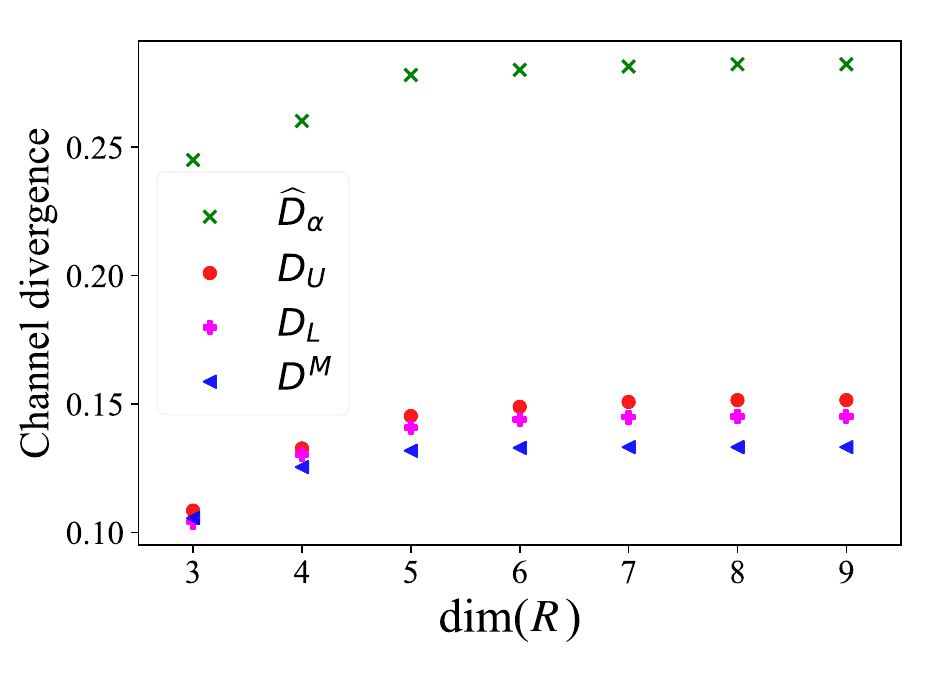}
  \caption{\label{fig:dim_rho_r} Channel divergences for two dephasing channels as a function of the Hilbert space truncation, for $E=1$; the dephasing parameters are $\gamma_1 = 0.1$, and $\gamma_2 = 0.4$.}
\end{figure}

To provide some insights into the optimal probe state for the dephasing channel, we examine the optimal state~$\rho$ generated by the SDP~\eqref{sdp:measured_sdp} -- note that the full input state is a purified version in the joint Hilbert space of the channel input and the reference system.
For $\gamma_1 = 0.1$, $\gamma_2 = 0.5$, and $E = 0.5$,  the optimal output state is 
$\rho_R = {0.77}|0\rangle\!\langle 0|+ {0.22}|2 \rangle\!\langle 2| + \Delta_R$, 
where $\Delta_R \geq0$, and supp$(\Delta_R) \in $span$\{\ket{n}: n\geq 3 \}$.
The input probe corresponds to its purification
$|\psi\rangle_{RA} = \sqrt{0.77}\ket{00}+ \sqrt{0.22}\ket{22}+  \sqrt{ \| \Delta_R \|_1}  \ket{\varphi_{n>2}}$, where $\ket{\varphi_{n>2}}$ is some state in $\operatorname{span}\{\ket{22}, \ket{33}, \ldots\}$. \color{black}. That is, the SDP selects a state that invests the largest filling fraction on the two-photon term.

In contrast, when the separation in the dephasing parameter is smaller ($ \gamma_1 = 0.1$ and $\gamma_2 = 0.15$), the optimal probe state is such that more weight is on the three-photon term, which is more sensitive to dephasing:  $\rho_R = {0.83}|0\rangle\!\langle 0|+ {0.03}|2 \rangle\!\langle 2|+ {0.11}|3 \rangle\!\langle 3| + \delta \Pi_{n>5}$.  This reflects the fact that discriminating closely spaced dephasing strengths requires accessing finer phase information, which the SDP accomplishes by allocating energy to higher-photon-number components.

Our computations are performed in a truncated Hilbert space with cutoff $\dim(R)$. We do not currently have an analytic truncation-error bound for the GRD of arbitrary loss-dephasing bosonic channels and obtaining such a bound remains an open problem. However, we can obtain such a bound in the case of two bosonic dephasing channels and provide a detailed proof of such a bound in Appendix~\ref{app:trunc-err-bnd}. As shown in Fig.~\ref{fig:dim_rho_r}, our results are indeed numerically stable with respect to increasing the truncation and the reference dimension used in the SDP, consistent with the analytical bound presented in App.~\ref{app:trunc-err-bnd}.  In particular, we use $E \leq 1$ and $\dim(R)=9$ in our calculations. In Fig.~\ref{fig:dim_rho_r}, we plot the channel divergences for $E = 1$ and observe that they do not change appreciably above dimension equal to $7$. We show a similar figure for the loss-dephasing channel in Appendix~\ref{app:trunc-plot}.

\section{Outlook}

\label{sec:conclusion}

Our results demonstrate that only a small number of photons are required to achieve meaningful discrimination gains. This has immediate implications for photonic quantum communication, frequency metrology, and quantum reading protocols, where limiting the photon number of the input state is crucial.

Future work includes generalising our results to multiple channel discrimination, as well as to composite hypothesis testing. The SDP developed in App.~\ref{app:GRD-SDP}  will be useful in obtaining bounds for the latter task (see~\eqref{eq:final-SDP-GRD-composite} in particular).

In this paper, we have focused on bounding error exponents for various settings of channel discrimination. Going forward, we are confident that our bounds here, combined with the techniques of~\cite{Nuradha2025} could lead to various bounds on the query complexity of energy-constrained channel discrimination, but we leave this investigation for future work. It would also be interesting to investigate bounds for energy-constrained channel discrimination in the symmetric setting, for which the approach of~\cite[Remark~28]{ji2025bary} might end up being applicable.

\begin{acknowledgments}
ZH is supported by an ARC DECRA Fellowship (DE230100144) ``Quantum-enabled super-resolution imaging'' and an RMIT Vice Chancellor's Senior Research Fellowship. LL acknowledges financial support from the European Union under the European Research Council (ERC Grant Agreement No.~101165230). VS thanks the Dieter Schwarz Exchange Programme on Quantum Communication and Security
at the Centre for Quantum Technologies for support. MMW acknowledges support from the National Science Foundation under Grant
No.~2329662. 
\end{acknowledgments}

%

\onecolumngrid

\appendix


\section{Quantum relative entropies}

\label{app:q-rel-ents}

The quantum relative entropy, also known as the Umegaki relative entropy~\cite{Ume62}, is defined for a pair of states $\rho$ and $\sigma$ as follows:
\begin{align}
    D(\rho\Vert\sigma) \coloneqq  \left\{
        \begin{array}{l l}
        \tr[\rho (\ln  \rho - \ln  \sigma)] \qquad & \text{if supp} (\rho ) \subseteq \text{supp} (\sigma), \\
        +\infty & \text{otherwise.}
        \end{array}\right.
\end{align}
The $\alpha$-geometric R\'enyi relative entropy of a pair of states is defined as follows~\cite{matsumoto2015new,Matsumoto2018}:
\begin{equation}\label{eq:geo_rel_ent_st_defn}
    \widehat{D}_{\alpha}\!\left(\rho\Vert\sigma\right) \coloneqq \left\{\begin{array}{l l}
        \frac{1}{\alpha - 1}\ln \operatorname{Tr}\!\left[G_{1-\alpha}\!\left(\rho,\sigma\right)\right] & \text{if supp}(\rho)\subseteq\text{supp}(\sigma), \\
        +\infty & \text{otherwise}
    \end{array}\right. \qquad \forall \alpha \in (1,2],
\end{equation}
where $G_{1-\alpha}(\rho,\sigma)$ is the weighted operator geometric mean defined in~\eqref{eq:geo_mean_defn}. 

The geometric R\'enyi relative entropy is monotonically non-decreasing with respect to $\alpha$~\cite[Prop.~72]{KatariyaWilde2021}, and it is a well-behaved divergence for $\alpha \in (0,1)$ as well~\cite{matsumoto2015new, Matsumoto2018,KatariyaWilde2021}. However, we restrict our development to $\alpha\in (1,2]$ for the purpose of this work. In the limit $\alpha \to 1$, the $\alpha$-geometric R\'enyi relative entropy converges to the Belavkin--Staszewski relative entropy~\cite[Prop.~79]{KatariyaWilde2021}, the latter defined as follows~\cite{Belavkin1982}:
\begin{equation}
    \widehat{D}\!\left(\rho\Vert\sigma\right) \coloneqq \operatorname{Tr}\!\left[D_{\operatorname{op}}(\rho\Vert\sigma)\right],
\end{equation}
where $D_{\operatorname{op}}(\cdot\Vert\cdot)$ is the operator relative entropy defined in~\eqref{eq:op_rel_ent_defn}. 
Note that
\begin{align}
\label{eq:associate}
D_{\operatorname{op}}(X_1\otimes X_2\| Y_1\otimes Y_2) = D_{\operatorname{op}}(X_1\|Y_1)\otimes X_2 + X_1 D_{\operatorname{op}}(X_2\|Y_2)
\end{align}
holds for all operators $X_1,X_2\geq 0$ and $Y_1,Y_2 > 0$. Also, in Ref.~\cite[Appendix A]{FangFawzi2021}, it was proved that
\begin{align}
\label{eq:mop}
D_{\operatorname{op}}(MXM^\dag\| MYM^\dag) \leq M D_{\operatorname{op}} (X\|Y) M^\dag
\end{align}
for all $X \geq 0, Y>0$, and all $M$; furthermore, Eq.~\eqref{eq:mop} holds with equality if $M$ is invertible.

The notion of relative entropy has been extended to quantum channels~\cite{cooney2016strong,LKDW18}. The Umegaki relative entropy of quantum channels is defined as follows:
\begin{equation}\label{eq:rel_ent_ch_defn}
    D\!\left(\mathcal{N}_{A\to B}\Vert\mathcal{M}_{A\to B}\right) \coloneqq \sup_{\rho_{RA}}D\!\left(\mathcal{N}_{A\to B}\!\left(\rho_{RA}\right)\middle\Vert \mathcal{M}_{A\to B}\!\left(\rho_{RA}\right)\right),
\end{equation}
where $R$ is a reference system of arbitrary dimension and the supremum is over every state $\rho_{RA}$. Similarly, the $\alpha$-geometric R\'enyi relative entropy of channels is defined as follows:
\begin{equation}\label{eq:geo_ent_ch_defn}
    \widehat{D}_{\alpha}\!\left(\mathcal{N}_{A\to B}\Vert\mathcal{M}_{A\to B}\right) \coloneqq \sup_{\rho_{RA}}\widehat{D}_{\alpha}\!\left(\mathcal{N}_{A\to B}\!\left(\rho_{RA}\right)\middle\Vert \mathcal{M}_{A\to B}\!\left(\rho_{RA}\right)\right) \qquad \forall \alpha \in (1,2].
\end{equation}

In both~\eqref{eq:rel_ent_ch_defn} and~\eqref{eq:geo_ent_ch_defn}, it suffices to restrict the supremum to all pure states with the dimension of the reference system $R$ equal to the dimension of the system $A$ (see~\cite[Section II.C]{LKDW18}). This leads to the following expression for the $\alpha$-geometric R\'enyi relative entropy of channels~\cite[Equation (31)]{FangFawzi2021}:
\begin{equation}\label{eq:geo_ch_div_Choi_mean}
    \widehat{D}_{\alpha}\!\left(\mathcal{N}_{A\to B}\middle\Vert\mathcal{M}_{A\to B}\right) = \sup_{\rho_R}\frac{1}{\alpha - 1}\ln \!\left(\operatorname{Tr}\!\left[\left(\rho_R\otimes \openone_B\right)G_{1-\alpha}\!\left(J^{\mathcal{N}}_{RB},J^{\mathcal{M}}_{RB}\right)\right]\right),
\end{equation}
where the supremum is over every state $\rho_R$ and 
\begin{equation}
\label{eq:Choi_op_defn}
    J^{\mathcal{N}}_{RB} \coloneqq \sum_{i,j=0}^{|A|-1}|i\rangle\!\langle j|_{R}\otimes \mathcal{N}_{A\to B}\!\left(|i\rangle\!\langle j|_{A}\right)
\end{equation}
is the Choi operator of the channel $\mathcal{N}_{A\to B}$ with system $A$ isomorphic to system $R$. The Choi operator of a channel is often expressed using the unnormalized maximally entangled operator, $\Gamma_{RA} \coloneqq \sum_{i,j=0}^{|A|-1}|i\rangle\!\langle j|_R\otimes |i\rangle\!\langle j|_{A}$, as follows:
\begin{equation}\label{eq:Choi_op_max_ent_rep}
    J^{\mathcal{N}}_{RB} = \mathcal{N}_{A\to B}\!\left(\Gamma_{RA}\right).
\end{equation}

In the limit $\alpha \to 1$, the $\alpha$-geometric R\'enyi relative entropy of channels converges to the Belavkin--Staszewski relative entropy of channels~\cite{FangFawzi2021}, which is defined as follows:
\begin{equation}
    \widehat{D}\!\left(\mathcal{N}_{A\to B}\middle\Vert \mathcal{M}_{A\to B}\right) \coloneqq \sup_{\rho_{RA}}\widehat{D}\!\left(\mathcal{N}_{A\to B}\!\left(\rho_{RA}\right)\middle\Vert\mathcal{M}_{A\to B}\!\left(\rho_{RA}\right)\right).
\end{equation}
Using~\eqref{eq:mop} and following the proof of~\eqref{eq:geo_ch_div_Choi_mean} (see Appendix~\ref{sec:proof1_grd} for the derivation with energy constraints), one can obtain the following equality:
\begin{equation}
    \widehat{D}\!\left(\mathcal{N}_{A\to B}\middle\Vert \mathcal{M}_{A\to B}\right) = \sup_{\rho_R}\operatorname{Tr}\!\left[\left(\rho_R\otimes \openone_B\right)D_{\operatorname{op}}\!\left(J^{\mathcal{N}}_{RB}\middle \|J^{\mathcal{M}}_{RB}\right)\right],
    \label{eq:BS-rel-ent-D-op}
\end{equation}
where the supremum is over every full rank state $\rho_R$.

\section{Proof of Equation~\texorpdfstring{\eqref{eq:EC_geo_ch}}{(10)}}

\label{sec:proof1_grd}

In this appendix, we provide a proof of Eq.~\eqref{eq:EC_geo_ch}. To begin with, let us note that
the energy-constrained $\alpha$-geometric R\'enyi relative entropy of channels is defined as follows~\cite[Eq.~(12.12)]{sharma2018bounding}:
\begin{equation}
    \widehat{D}_{\alpha,H,E}\!\left(\mathcal{N}_{A\to B}\middle\Vert\mathcal{M}_{A\to B}\right) \coloneqq \sup_{\substack{\psi_{RA},\\ \operatorname{Tr}[H_A\psi_A] \le E}}\widehat{D}_{\alpha}\!\left(\mathcal{N}_{A\to B}\!\left(\psi_{RA}\right)\middle\Vert\mathcal{M}_{A\to B}\!\left(\psi_{RA}\right)\right),
\end{equation}
where $H_A$ is the Hamiltonian associated with the system and $E$ is the energy constraint. The supremum is over every pure state $\psi_{RA}$, and $\psi_A$ refers to the marginal of the state $\psi_{RA}$ on the system $A$. 

Let $\left\{|i\rangle_A\right\}_{i=0}^{\infty}$ be the eigenbasis of the Hamiltonian $H_A$. Since $\psi_{RA}$ is a purification of $\psi_A$, there exists an isometry $V_{R'\to R}$ such that 
\begin{align}
    |\psi\rangle_{RA} &= \left(V_{R'\to R}\otimes\sqrt{\psi}_A\right)\sum_{i=0}^{\infty}|i\rangle_{R'}|i\rangle_A\\
    &= \left(V_{R'\to R}\sqrt{\psi^T}_{R'}\otimes \openone_A\right)\sum_{i=0}^{\infty}|i\rangle_{R'}|i\rangle_A,
\end{align}
where the second equality follows from the transpose trick \cite[Eq.~(2.2.40)]{khatri2020principles}. As such,
\begin{align}
    \psi_{RA} &= |\psi\rangle\!\langle \psi|_{RA}\\
    &= V_{R'\to R}\sqrt{\psi^T}_{R'}\left(\sum_{i,j=0}^{\infty}|i\rangle\!\langle j|_{R'}\otimes |i\rangle\!\langle j|_A\right)\sqrt{\psi^T}_{R'}\left(V_{R'\to R}\right)^{\dagger}\\
    &= V_{R'\to R}\sqrt{\psi^T}_{R'}\Gamma_{R'A}\sqrt{\psi^T}_{R'}\left(V_{R'\to R}\right)^{\dagger},
\end{align}
where $\Gamma_{R'A}$ is the unnormalized maximally entangled operator defined in~\eqref{eq:Choi_op_max_ent_rep}. Note that the transpose is taken with respect to the eigenbasis of $H_A$. The energy-constrained $\alpha$-geometric R\'enyi relative entropy can thus be expressed as follows:
\begin{align}
    &\widehat{D}_{\alpha,H,E}\!\left(\mathcal{N}_{A\to B}\middle\Vert\mathcal{M}_{A\to B}\right)\notag\\ 
    &= \sup_{\substack{\psi_{R'}\ge 0, \operatorname{Tr}[\psi_{R'}] = 1,\\ \operatorname{Tr}[H_{R'}\psi_{R'}]\le E}}\widehat{D}_{\alpha}\!\left(\mathcal{N}\!\left(V_{R'\to R}\sqrt{\psi^T}_{R'}\Gamma_{R'A}\sqrt{\psi^T}_{R'}\left(V_{R'\to R}\right)^{\dagger}\right)\middle\Vert\mathcal{M}\!\left(V_{R'\to R}\sqrt{\psi^T}_{R'}\Gamma_{R'A}\sqrt{\psi^T}_{R'}\left(V_{R'\to R}\right)^{\dagger}\right)\right)\\
    &= \sup_{\substack{\psi_{R'}\ge 0, \operatorname{Tr}[\psi_{R'}] = 1,\\ \operatorname{Tr}[H_{R'}\psi_{R'}]\le E}}\widehat{D}_{\alpha}\!\left(V_{R'\to R}\sqrt{\psi^T}_{R'}\mathcal{N}\!\left(\Gamma_{R'A}\right)\sqrt{\psi^T}_{R'}\left(V_{R'\to R}\right)^{\dagger}\middle\Vert V_{R'\to R}\sqrt{\psi^T}_{R'}\mathcal{M}\!\left(\Gamma_{R'A}\right)\sqrt{\psi^T}_{R'}\left(V_{R'\to R}\right)^{\dagger}\right)\\
    &= \sup_{\substack{\psi_{R'}\ge 0, \operatorname{Tr}[\psi_{R'}] = 1,\\ \operatorname{Tr}[H_{R'}\psi_{R'}]\le E}}\widehat{D}_{\alpha}\!\left(\sqrt{\psi^T}_{R'}J^{\mathcal{N}}_{R'B}\sqrt{\psi^T}_{R'}\middle\Vert\sqrt{\psi^T}_{R'}J^{\mathcal{M}}_{R'B}\sqrt{\psi^T}_{R'}\right),\label{eq:EC_geo_rel_ent_sup_psd_states}
\end{align}
where the second equality follows from the fact that the channels $\mathcal{N}$ and $\mathcal{M}$ only act on system $A$, and the final equality follows from the isometric invariance of the $\alpha$-geometric R\'enyi relative entropy along with~\eqref{eq:Choi_op_max_ent_rep}. The set of positive definite states with energy $E$ is dense in the set of all states with energy $E$ if $E$ is strictly larger than the minimum eigenvalue of the Hamiltonian $H_A$. For all further discussions, we assume that the energy budget $E$ is strictly larger than the minimum eigenvalue of $H_A$. The $\alpha$-geometric R\'enyi relative entropy of states satisfies desirable continuity properties \cite[App.~A]{Mosonyi2024}. Therefore, we can restrict the supremum in~\eqref{eq:EC_geo_rel_ent_sup_psd_states} to positive definite states with energy  $\leq E$. That is,
\begin{equation}
\label{eq:D10} \widehat{D}_{\alpha,H,E}\!\left(\mathcal{N}_{A\to B}\middle\Vert\mathcal{M}_{A\to B}\right) = \sup_{\substack{\psi_{R'} > 0, \operatorname{Tr}[\psi_{R'}] = 1,\\ \operatorname{Tr}[H_{R'}\psi_{R'}]\le E}}\widehat{D}_{\alpha}\!\left(\sqrt{\psi^T}_{R'}J^{\mathcal{N}}_{R'B}\sqrt{\psi^T}_{R'}\middle\Vert\sqrt{\psi^T}_{R'}J^{\mathcal{M}}_{R'B}\sqrt{\psi^T}_{R'}\right).
\end{equation}
Let us assume that $\operatorname{supp}\!\left(J^{\mathcal{N}}_{R'B}\right)\subseteq \operatorname{supp}\!\left(J^{\mathcal{M}}_{R'B}\right)$. Recalling the definition of $\alpha$-geometric R\'enyi relative entropy of states from~\eqref{eq:geo_rel_ent_st_defn}, we can write
\begin{align}
    \widehat{D}_{\alpha,H,E}\!\left(\mathcal{N}_{A\to B}\middle\Vert\mathcal{M}_{A\to B}\right) &= \sup_{\substack{\psi_{R'} > 0, \operatorname{Tr}[\psi_{R'}] = 1,\\ \operatorname{Tr}[H_{R'}\psi_{R'}]\le E}}\frac{1}{\alpha - 1}\ln\operatorname{Tr}\!\left[G_{1-\alpha}\!\left(\sqrt{\psi^T}_{R'}J^{\mathcal{N}}_{R'B}\sqrt{\psi^T}_{R'},\sqrt{\psi^T}_{R'}J^{\mathcal{M}}_{R'B}\sqrt{\psi^T}_{R'}\right)\right]\\
    &= \sup_{\substack{\psi_{R'} > 0, \operatorname{Tr}[\psi_{R'}] = 1,\\ \operatorname{Tr}[H_{R'}\psi_{R'}]\le E}}\frac{1}{\alpha - 1}\ln\operatorname{Tr}\!\left[\sqrt{\psi^T}_{R'}G_{1-\alpha}\!\left(J^{\mathcal{N}}_{R'B},J^{\mathcal{M}}_{R'B}\right)\sqrt{\psi^T}_{R'}\right]\\
    &= \sup_{\substack{\psi_{R'} > 0, \operatorname{Tr}[\psi_{R'}] = 1,\\ \operatorname{Tr}[H_{R'}\psi_{R'}]\le E}}\frac{1}{\alpha - 1}\ln\operatorname{Tr}\!\left[\left(\psi^T_{R'}\otimes I_B\right)G_{1-\alpha}\!\left(J^{\mathcal{N}}_{R'B},J^{\mathcal{M}}_{R'B}\right)\right]\\
    &= \sup_{\substack{\psi_{R'} > 0, \operatorname{Tr}[\psi_{R'}] = 1,\\ \operatorname{Tr}[H_{R'}\psi_{R'}]\le E}}\frac{1}{\alpha - 1}\ln\operatorname{Tr}\!\left[\psi^T_{R'}\operatorname{Tr}_B\!\left[G_{1-\alpha}\!\left(J^{\mathcal{N}}_{R'B},J^{\mathcal{M}}_{R'B}\right)\right]\right],
\end{align}
where the second equality follows from the transformer equality given in~\eqref{eq:mop} and the penultimate equality follows from the cyclicity of trace. 
Let us define $\rho_{R'} \coloneqq \psi^T_{R'}$. Then,
\begin{align}
    \widehat{D}_{\alpha,H,E}\!\left(\mathcal{N}_{A\to B}\middle\Vert\mathcal{M}_{A\to B}\right) &=\sup_{\rho^T_{R'} > 0, \operatorname{Tr}[\rho^T_{R'}] = 1}\left\{\frac{1}{\alpha - 1}\ln\operatorname{Tr}\!\left[\rho_{R'}\operatorname{Tr}_B\!\left[G_{1-\alpha}\!\left(J^{\mathcal{N}}_{R'B},J^{\mathcal{M}}_{R'B}\right)\right]\right]: \operatorname{Tr}\!\left[H_{R'}\rho^T_{R'}\right] \le E\right\}\\
    &=\sup_{\rho^T_{R'} > 0, \operatorname{Tr}[\rho^T_{R'}] = 1}\left\{\frac{1}{\alpha - 1}\ln\operatorname{Tr}\!\left[\rho_{R'}\operatorname{Tr}_B\!\left[G_{1-\alpha}\!\left(J^{\mathcal{N}}_{R'B},J^{\mathcal{M}}_{R'B}\right)\right]\right]: \operatorname{Tr}\!\left[H^T_{R'}\rho_{R'}\right] \le E\right\}\\
    &=\sup_{\rho^T_{R'} > 0, \operatorname{Tr}[\rho^T_{R'}] = 1}\left\{\frac{1}{\alpha - 1}\ln\operatorname{Tr}\!\left[\rho_{R'}\operatorname{Tr}_B\!\left[G_{1-\alpha}\!\left(J^{\mathcal{N}}_{R'B},J^{\mathcal{M}}_{R'B}\right)\right]\right]: \operatorname{Tr}\!\left[H_{R'}\rho_{R'}\right] \le E\right\}\\
    &=\sup_{\rho_{R'} > 0, \operatorname{Tr}[\rho_{R'}] = 1}\left\{\frac{1}{\alpha - 1}\ln\operatorname{Tr}\!\left[\rho_{R'}\operatorname{Tr}_B\!\left[G_{1-\alpha}\!\left(J^{\mathcal{N}}_{R'B},J^{\mathcal{M}}_{R'B}\right)\right]\right]: \operatorname{Tr}\!\left[H_{R'}\rho_{R'}\right] \le E\right\},\label{eq:EC_geo_ch_final}
\end{align}
where the second equality follows from the fact that the transpose map is self-adjoint, the third equality follows from the fact that $H^T_{R'} = H_{R'}$ because the transpose is with respect to the eigenbasis of the Hamiltonian, and the final equality follows by noting that there exists a unique state $\rho^T_{R'}$ for every state $\rho_{R'}$. Finally, one can relax the positive definite constraint on $\rho_{R'}$ to a positive semidefinite constraint because the objective function of the supremum in~\eqref{eq:EC_geo_ch_final} is continuous with respect to $\rho_{R'}$.

\subsection{Closed form for the energy-constrained Belavkin--Staszewski channel divergence with energy constraint}

Consider now a finite-dimensional quantum system $A$ endowed with a grounded Hamiltonian $H_A$ (i.e., such that the minimum eigenvalue of $H_A$ is equal to zero).
For two channels $\mathcal{N}_{A\to B}$ and $\mathcal{M}_{A\to B}$ and $E > 0$, we define the \emph{energy-constrained Belavkin--Staszewski} (ECBS) channel divergence as follows:
\begin{equation}
\widehat{D}_{H,E}(\mathcal{N}\|\mathcal{M}) \coloneqq  \sup_{\substack{\psi_{RA}, \\ \mathrm{Tr}\left[\psi_A H_A\right] \le E}} \widehat{D}\!\left( \mathcal{N}_{A\to B}(\psi_{RA}) \middle\| \mathcal{M}_{A\to B}(\psi_{RA}) \right), \label{eq:def-ecbs}
\end{equation}
where the supremum is over every bipartite pure state $\psi_{RA}$ such that $\operatorname{Tr}\left[\psi_A H_A\right] \le E$. Similar to the case of energy-constrained $\alpha$-geometric R\'enyi relative entropy of channels, we can assume that system $R$ is isomorphic to system $A$ without loss of generality.

\vspace{3mm}
\noindent
\begin{lemma}\label{lem:ECBS_op_rel_ent_form}
Let $\mathcal{N},\mathcal{M}$ be two quantum channels. It holds that
\begin{equation}
\widehat{D}_{H,E}(\mathcal{N}\|\mathcal{M}) = \sup_{\substack{\rho_R\ge 0, \operatorname{Tr}[\rho_R] = 1,\\ \mathrm{Tr}\left[\rho_R H_R\right] \le E}} \mathrm{Tr} \left[ \rho_R\, D_{\mathrm{op}}\!\left( J^{\mathcal{N}}_{RB}\middle \| J^{\mathcal{M}}_{RB} \right) \right], \label{eq:closed-form}
\end{equation}
where $J^{\mathcal{N}}_{RB}$ and $J^{\mathcal{M}}_{RB}$ are the Choi operators of channels $\mathcal{N}_{A\to B}$ and $\mathcal{M}_{A\to B}$, respectively.
\end{lemma}

\begin{proof}
The proof here is similar to the proof of~\eqref{eq:EC_geo_ch}, replacing the transformer equality from~\eqref{eq:transformer_eq} with~\eqref{eq:mop}.
The argument is also similar to the proof of~\cite[Lemma 5]{FangFawzi2021}.  
We write
\begin{align}
\widehat{D}_{H,E}(\mathcal{N}\|\mathcal{M}) 
&= \sup_{\rho:\, \mathrm{Tr}[\rho_A H_A] \le E} \widehat{D}\!\left( \mathcal{N}_{A\to B}(\sqrt{\rho_A} \Gamma_{RA} \sqrt{\rho_A}) \middle\| \mathcal{M}_{A\to B}(\sqrt{\rho_A} \Gamma_{RA} \sqrt{\rho_A}) \right) \\
&= \sup_{\rho:\, \mathrm{Tr}[\rho_A H_A] \le E} \widehat{D}\!\left( \sqrt{\rho^T_R} J^{\mathcal{N}}_{RB} \sqrt{\rho^T_R} \middle\| \sqrt{\rho^T_R} J^{\mathcal{M}}_{RB} \sqrt{\rho^T_R} \right) \\
&\le \sup_{\rho:\, \mathrm{Tr}[\rho_A H_A] \le E} \mathrm{Tr}_{RB} \left[ \rho^T_R\, D_{\mathrm{op}}\!\left( J^{\mathcal{N}}_{RB}\middle \| J^{\mathcal{M}}_{RB} \right) \right] \\
&= \sup_{\rho:\, \mathrm{Tr}[\rho_R H_R] \le E} \mathrm{Tr}_{RB} \left[ \rho_R\, D_{\mathrm{op}}\!\left( J^{\mathcal{N}}_{RB}\middle\| J^{\mathcal{M}}_{RB} \right) \right],
\end{align}
where the steps follow from:
\begin{enumerate}
\item Decomposition of $\ket{\Psi}_{RA}$ into $(U_R \otimes \sqrt{\rho_A})\ket{\Gamma}_{RA}$ and simplifying $U_R$.
\item The identity $(M_A \otimes I_R)\ket{\Gamma}_{RA} = (I_A \otimes M^T_R)\ket{\Gamma}_{RA}$.
\item The transformer inequality~\eqref{eq:mop}.
\item $H_A^T = H_A$ by our basis choice, and renaming $A\to R$.
\end{enumerate}
The reverse inequality follows similarly using the equality case of~\eqref{eq:mop} and continuity.

To conclude the proof, we also observe that
\begin{align}
\widehat D_{H,E} (\curlN \| \curlM) &\geq \sup_{\rho:\, \Tr[\rho_A H_A] \leq E} 
\widehat D(\curlN_{A\rightarrow B}( \sqrt{\rho_A}\Gamma_{RA} \sqrt{\rho_A}) \| \curlM_{A\rightarrow B}( \sqrt{\rho_A}\Gamma_{RA} \sqrt{\rho_A}  ) \\
&=  \sup_{\rho:\, \Tr[\rho_A H_A] \leq E} \widehat D\!\left( \sqrt{\rho_R^T} J^\curlN_{RB}  \sqrt{\rho_R^T} \middle\| \sqrt{\rho_R^T} J^\curlM_{RB}  \sqrt{\rho_R^T}  \right) \\
&=  \sup_{\rho:\, \Tr[\rho_A H_A] \leq E} \tr_{AB} [\rho_R^T D_{\operatorname{op}}(J^\curlN_{RB} \| J^\curlM_{RB})]  \label{step:5}\\
&= \sup_{\rho:\, \Tr[\rho_R H_R] \leq E} \tr_{RB} [\rho_R D_{\operatorname{op}}(J^\curlN_{RB} \| J^\curlM_{RB})]  \label{step:6} ,
\end{align}
where~\eqref{step:5} follows from the equality case of~\eqref{eq:mop}, and for~\eqref{step:6}, we used the continuity of $\rho\rightarrow \tr[\rho_A D_{\operatorname{op}}(J^\curlN_{AB} \| J^\curlM_{AB})].$
\end{proof}

\subsection{Chain rule for the energy-constrained Belavkin--Staszewski channel divergence}
\label{eq:sec_ECBS}
\begin{lemma}
Let $A$ be a quantum system endowed with a grounded Hamiltonian $H_A$.  
Let $\rho_{RA}$ be a bipartite state with $\mathrm{Tr}\!\left[\rho_A H_A\right] \le E$ for some $E > 0$. Let $\mathcal{N}_{A\to B}$ and $\mathcal{M}_{A\to B}$ be arbitrary channels. The following inequality holds for every state $\sigma_{RA}$:
\begin{equation}
\widehat{D}\!\left( \mathcal{N}_{A\to B}(\rho_{RA}) \middle \| \mathcal{M}_{A\to B}(\sigma_{RA}) \right) \le \widehat{D}(\rho_{RA} \| \sigma_{RA}) + \widehat{D}_{H,E}(\mathcal{N}\|\mathcal{M}). \label{eq:chain-rule}
\end{equation}
\end{lemma}

\begin{proof}
The action of a channel $\mathcal{N}_{A\to B}$ on a state $\rho_{RA}$ is given in terms of the Choi operator of the channel as follows:
\begin{equation}\label{eq:ch_action_choi}
    \mathcal{N}_{A\to B}\!\left(\rho_{RA}\right) = \langle \Gamma|_{SA}\rho_{RA}\otimes J^{\mathcal{N}}_{SB}|\Gamma\rangle_{SA},
\end{equation}
where system $S$ is isomorphic to system $R$ and $|\Gamma\rangle_{SA}$ is the unnormalized maximally entangled  vector. Also, recall the well-known identity:
\begin{equation}\label{eq:sandwich_max_ent_eq_tr}
    \operatorname{Tr}_A\!\left[X_{RA}\right] = \langle \Gamma|_{AB}X_{RA}\otimes \openone_B|\Gamma\rangle_{AB}.
\end{equation}
Let $\rho_{RA}$ be a bipartite state with $\operatorname{Tr}\!\left[\rho_A H_A\right] \le E$ for some $E>0$. Now consider that
\begin{align}
  &\widehat{D}\!\left(\mathcal N_{A\to B}(\rho_{RA}) \middle\Vert \mathcal M_{A\to B}(\sigma_{RA})\right)\notag \\ 
  &= \Tr\!\left[D_{\operatorname{op}} \! \left(\mathcal N_{A\to B}(\rho_{RA}) \| \mathcal M_{A\to B}(\sigma_{RA})\right)\right]
  \\
  &= \Tr\!\left[ D_{\operatorname{op}}\!\left(
      \langle \Gamma|_{SA}\rho_{RA}\otimes J^{\mathcal{N}}_{SB}|\Gamma\rangle_{SA}\middle \|
      \langle \Gamma|_{SA}\sigma_{RA}\otimes J^{\mathcal{M}}_{SB}|\Gamma\rangle_{SA}
    \right)\right] \label{eq:chain_rule_step_1}
  \\
  &\le \Tr\!\left[\langle \Gamma |_{SA} 
      D_{\operatorname{op}}\!\left(\rho_{RA}\otimes J^{\mathcal N}_{SB}\middle\| \sigma_{RA}\otimes J^{\mathcal M}_{SB}\right)
      | \Gamma\rangle_{SA}\right] \label{eq:chain_rule_step_2}
  \\
  &= \Tr\!\left[ \langle \Gamma |_{SA}
      \left(
        D_{\operatorname{op}}\!\left(\rho_{RA}\|\sigma_{RA}\right)\otimes J^{\mathcal N}_{SB}
        + \rho_{RA}\otimes D_{\operatorname{op}}\!\left(J^{\mathcal N}_{SB}\middle \| J^{\mathcal M}_{SB}\right)
      \right)
      | \Gamma\rangle_{SA}\right]  \label{eq:chain_rule_step_3}\\
  &= \langle \Gamma |_{SA} 
      \!\left(
        \Tr_{RB}\!\left[D_{\operatorname{op}}(\rho_{RA}\|\sigma_{RA})\otimes J^{\mathcal N}_{SB}\right]
        + \Tr_{RB}\!\left[\rho_{RA}\otimes D_{\operatorname{op}}(J^{\mathcal N}_{SB}\| J^{\mathcal M}_{SB})\right]
      \right)
      | \Gamma\rangle_{SA}  \label{eq:chain_rule_step_4}
  \\
  &= \langle \Gamma |_{SA}
      \!\left(
        \Tr_{R}\!\left[D_{\operatorname{op}}(\rho_{RA}\|\sigma_{RA})\right]\otimes \openone_S
        + \rho_A\otimes \Tr_{B}\!\left[D_{\operatorname{op}}(J^{\mathcal N}_{SB}\| J^{\mathcal M}_{SB})\right]
      \right)
      | \Gamma\rangle_{SA} \label{eq:chain_rule_step_5}
  \\
  &= \Tr\!\left[D_{\operatorname{op}}(\rho_{RA}\|\sigma_{RA})\right]
     + \langle \Gamma|_{SA} \openone_A\otimes \rho^T_S D_{\operatorname{op}}(J^{\mathcal N}_{SB}\| J^{\mathcal M}_{SB})|\Gamma\rangle_{SA}
  \label{eq:chain_rule_step_6} \\
  &= \widehat D (\rho_{RA}\Vert \sigma_{RA}) + \operatorname{Tr}\!\left[\rho^T_S D_{\operatorname{op}}(J^{\mathcal N}_{SB}\| J^{\mathcal M}_{SB})\right]\label{eq:BS_decomp_ch_div_st_div}\\
  &\le \widehat D (\rho_{RA}\Vert \sigma_{RA}) + \widehat{D}_{H,E}(\mathcal N \Vert \mathcal M),
\end{align}
where the first inequality follows from~\eqref{eq:mop}, the third equality follows from~\eqref{eq:associate}, the fifth equality follows from the fact that $\operatorname{Tr}_B\!\left[J^{\mathcal{L}}_{AB}\right] = \openone_A$ for every channel $\mathcal{L}_{A\to B}$, the penultimate equality follows from~\eqref{eq:sandwich_max_ent_eq_tr} and the transpose trick, and the final inequality follows from Lemma~\ref{lem:ECBS_op_rel_ent_form}.
\end{proof}

\section{Weighted operator geometric mean and its properties}

\label{app:weighted-op-geo-mean}

The weighted operator geometric mean of two positive definite operators, $A$ and $B$, is defined for a weight $t\in \mathbb{R}$ as follows~\cite{pusz1975functional,Kubo1980}:
\begin{equation}
    G_t(A,B) = A^{\frac{1}{2}}\!\left(A^{-\frac{1}{2}}BA^{-\frac{1}{2}}\right)^tA^{\frac{1}{2}}.
\end{equation}
It can be extended to positive semidefinite operators by taking the inverse with respect to the support of $A$. 
The weighted operator geometric mean $G_t$ has the following properties (see \cite[Eq.~(7.6.5)]{khatri2020principles} and~\cite{Kubo1980}):
\begin{align}
    G_t(A,B) & = G_{1-t}(B,A) \qquad \forall t\in \mathbb{R},\\
\label{eq:wt_geo_mean_monotonic}
    A_1\le A_2, B_1\le B_2  \implies G_t(A_1,B_1) & \le G_t(A_2,B_2) \qquad \forall t\in [0,1]. 
\end{align}
Furthermore, for all $s,t\in \mathbb{R}$,
    \begin{align}
        G_{s}\!\left(A,G_t(A,B)\right) &= G_{st}(A,B),\label{eq:geo_wt_multiplicative}\\
        G_s\!\left(G_t(A,B),B\right) &= G_{s+t-st}(A,B).
    \end{align}
It satisfies the transformer inequality; i.e., for every linear operator $K$,
    \begin{equation}\label{eq:transformer_ineq}
        G_t(KAK^{\dagger},KBK^{\dagger}) \le KG_t(A,B)K^{\dagger} \qquad \forall t\in [-1,0).
    \end{equation}
    If $K$ is invertible, then
    \begin{equation}\label{eq:transformer_eq}
        G_t(KAK^{\dagger},KBK^{\dagger}) = KG_t(A,B)K^{\dagger} \qquad \forall t\in [-1,0).
    \end{equation}

The weighted operator geometric mean is intricately linked to the relative operator entropy, the latter of which is defined as follows~\cite{fujii1989relative}:
\begin{align}
    D_{\operatorname{op}}(A\Vert B) &\coloneqq -A^{\frac{1}{2}}\ln \!\left(A^{-\frac{1}{2}}BA^{-\frac{1}{2}}\right)A^{\frac{1}{2}}\label{eq:op_rel_ent_defn}\\
    &= A^{\frac{1}{2}}\ln \!\left(A^{\frac{1}{2}}B^{-1}A^{\frac{1}{2}}\right)A^{\frac{1}{2}}.
\end{align}
In particular,
\begin{equation}
    \left.\frac{d}{dt}G_t\!\left(A,B\right)\right|_{t=0} = - D_{\operatorname{op}}\!\left(A\Vert B\right),
\end{equation}
which follows from the fact that $\left. \frac{d}{dt}x^t\right|_{t=0} = \ln  x$.

\section{Quantum relative entropy of channels}

\label{app:q-rel-ent-channels}

Recall the definition of relative entropy of channels from~\eqref{eq:rel_ent_ch_defn}. 
The authors of Ref.~\cite{kossmann2024}  obtained a semidefinite program that computes a lower bound on the energy-constrained relative entropy of channels, which is tight up to an error $\varepsilon$. Let 
\begin{equation}
    c^{\star}\!\left(\mathcal{N}_{A\to B},\mathcal{M}_{A\to B}\right) \coloneqq \sup_{\substack{\rho_R\ge 0,\\ \left(Q_k\right)_{k=1}^r}} 
\left\{
\begin{array}{c}
     \tr[\rho_R  (J^\curlN_{RB}-J^\curlM_{RB} )]
 + \sum_{k=1}^r \tr[Q_k (\alpha_k J^\curlN_{RB} +\beta_k J^\curlM_{RB} )] + \ln \lambda+1 -\lambda :\\
 \tr[\rho_A] =1, \quad \operatorname{Tr}\!\left[H_A\rho_A\right] \le E, \quad 0\leq Q_k \leq \rho_R \otimes \openone_B 
\end{array}\right\}.
\end{equation}
Then~\cite[Theorem III.1]{kossmann2024} states that
\begin{equation}
0 \le |c^{\star}\!\left(\mathcal{N}_{A\to B},\mathcal{M}_{A\to B}\right) - D_{H,E}\!\left(\mathcal{N}_{A\to B}\middle\Vert\mathcal{M}_{A\to B}\right) | \le \varepsilon,
\end{equation}
where $\ln\lambda \coloneqq D_{\max}\!\left(J^{\mathcal{N}}\middle\Vert J^{\mathcal{M}}\right)$ and $r = O(\sqrt{\lambda/\varepsilon })$. Here $D_{\max}(\cdot\Vert\cdot)$ refers to the max-relative entropy of channels defined in~\cite[Lemma~12]{wilde2020amortized} and $\lambda$ itself can be computed via the following SDP:
\begin{equation}
    \lambda = \inf\!\left\{\mu : J^{\mathcal{N}}_{RB}\le \mu J^{\mathcal{M}}_{RB} \right\}.
\end{equation}
The coefficients $\alpha_k,\beta_k$ can be calculated by the following: choose a discretisation 
\begin{align}
\mu = t_0 < t_1 <\cdots <t_r =\lambda.
\end{align}
The choice of $t_i$ can be made by employing~\cite[Eq.~(5)]{kossmann2025}. Then
\begin{align}
\alpha_k = \ln \! \left(\frac{t_k}{t_{k+1}}\right), \qquad \beta_k = t_{k+1} -t_k.
\end{align}

An upper bound approximation of the energy-constrained relative entropy of channels has also been derived~\cite{kossmann2024}. For $k\in \{1,\ldots,r\},~\gamma_0 = 1~,\delta_0 = -1$,
\begin{align}
\label{eq:upper}
\inf_{\substack{ x\in\mathbb{R},~y\geq 0 \\ N_0 \in \text{Herm}}} 
\left\{
\begin{array}{c}
x+y E + \ln \lambda+1-\lambda: \\
N_k \geq \gamma_k J^\curlN_{RB} + \delta_k J^\curlM_{RB},\quad
N_k \geq 0 \\
x \openone_R + y H_R \geq \tr_B \left[ \sum_{k=0}^r N_k \right]\\
\end{array}
\right\}
\end{align}
with coefficients $\gamma_k,\delta_k$ chosen as in~\cite[Corollary 1, Appendix E]{kossmann2025}.

\section{Measured relative entropy of states and channels}

\label{sec:measured_re}

The measured relative entropy between two states $\rho$ and $\sigma$ is defined as follows:
\begin{equation}
    D^M\!\left(\rho\middle\Vert\sigma\right) \coloneqq \sup_{\substack{\mathcal{X},\left(\Lambda^x\right)_{x\in \mathcal{X}}}} \left\{\operatorname{Tr}\!\left[\Lambda^x\rho\right]\ln\!\left(\frac{\operatorname{Tr}\!\left[\Lambda^x\rho\right]}{\operatorname{Tr}\!\left[\Lambda^x\sigma\right]}\right): \Lambda^x \ge 0 \  \forall x\in \mathcal{X}, \sum_{x\in \mathcal{X}}\Lambda^x = \openone\right\},
\end{equation}
where the supremum is over every finite set $\mathcal{X}$ and every POVM $\left(\Lambda^x\right)_{x\in \mathcal{X}}$. The following variational expression for the \emph{measured relative entropy} of states was reported in~\cite[Lem.~1 \& Thm.~2]{berta2017variational}:
\begin{equation}
D^M(\rho\|\sigma) = \sup_{\omega \geq 0} \{ \tr\!\left[(\ln \omega)\rho\right] - \tr[\omega \sigma] +1 \},
\end{equation}
which can alternatively be written as follows~\cite{huang2024semi}:
\begin{equation}
D^{M}(\rho\Vert\sigma)=\sup_{\omega >0, \theta \in \operatorname{Herm}}
\left\{  \tr[\theta\rho]-\operatorname{Tr}[\omega\sigma]+1~:~\theta\leq -D_{\operatorname{op}}
(\openone,\omega)\right\}  , 
\end{equation}
where 
$D_{\operatorname{op}}(\cdot\Vert\cdot)$ is the relative operator entropy defined in~\eqref{eq:op_rel_ent_defn}. 

In Ref.~\cite{fawzi2019semidefinite}, the authors obtained a set of semidefinite conditions that closely approximate the set $\left\{\theta: \theta \ge D_{\operatorname{op}}(X,Y)\right\}$ for any given pair of finite-dimensional positive definite operators $X$ and $Y$. The results of~\cite{fawzi2019semidefinite}, when combined with~\eqref{eq:K_cone}, yield a semidefinite program that approximates the measured relative entropy between two given states.

We now turn to the task at hand, which is to efficiently approximate the \emph{energy-constrained measured relative entropy} of channels. Let $A$ be a finite-dimensional quantum system endowed with a grounded Hamiltonian $H_A$. Then the energy-constrained measured relative entropy of channels is defined for a energy budget $E>0$ as follows:
\begin{equation}
    D^M_{H,E}\!\left(\mathcal{N}_{A\to B}\middle\Vert\mathcal{M}_{A\to B}\right) \coloneqq \sup_{\rho_{RA}, \operatorname{Tr}\left[\rho_AH_A\right] \le E} D^M\!\left(\mathcal{N}_{A\to B}\!\left(\rho_{RA}\right)\middle\Vert\mathcal{M}_{A\to B}\!\left(\rho_{RA}\right)\right),
\end{equation}
where the supremum is over every bipartite state $\rho_{RA}$ such that $\operatorname{Tr}\!\left[\rho_AH_A\right] \le E$. Similar to the case of the energy-constrained $\alpha$-geometric R\'enyi relative entropy of channels, we can assume that system $R$ is isomorphic to system $A$ without loss of generality.  In Ref.~\cite{huang2024semi},
the authors showed that the energy-constrained measured relative entropy of channels can be expressed in the following variational form:
\begin{equation}
\label{eq:G7}
D_{H,E}^{M}(\mathcal{N}_{A\to B}\Vert\mathcal{M}_{A\to B})=\sup_{\substack{\Omega,\rho\ge 0,\\ \Theta \in \text{Herm}}}
\left\{
\begin{array}
[c]{c}%
\operatorname{Tr}[\Theta_{RB} J^\curlN_{RB}]-\operatorname{Tr}[\Omega_{RB}
J^\curlM_{RB}]+1:\\
\operatorname{Tr}[\rho_R]=1,\quad\operatorname{Tr}[H_R\rho_R]\leq E,\\
\Theta_{RB}\leq -D_{\operatorname{op}}(\rho_R\otimes \openone _B,\Omega_{RB})
\end{array}
\right\}  .
\end{equation}
Now we can employ the results from Ref.~\cite{fawzi2019semidefinite} to obtain a semidefinite program that approximates the quantity in~\eqref{eq:G7}, which leads to the following SDP approximation of the energy-constrained measured relative entropy of channels:
\begin{equation}\label{eq:meas_ch_RE_SDP} 
D^M_{H,E}\!\left(\mathcal{N}\middle\Vert\mathcal{M}\right) \approx 
\sup_{\substack{\Omega \geq0,\rho\geq0, \\ \Theta \in \operatorname{Herm} } }
\left\{ 
\begin{array}
[c]{c}
\tr[\Theta J^\curlN_{RB}]-\Tr[\Omega J^\curlM_{RB}]+1:\\
\tr[\rho]=1,\quad 
\tr [H\rho]\leq E,\\ \\ 
T_1,\ldots,T_m,~Z_0,\ldots,Z_k \in \operatorname{Herm,}\\ 
X = \rho\otimes \openone, \quad Z_0 = \Omega,  \quad
\left \{ \left[
\begin{array}{cc}
 Z_i & Z_{i +1} \\
 Z_{i +1}  & X \\
\end{array}
\right] \geq 0 \right\} _{i = 0}^{k-1},
\\ \\
 \sum_{j=1}^m w_j T_j = 2^{-k} \Theta,\qquad
\left \{ \left[
\begin{array}{cc}
 Z_k-X-T_j & -\sqrt{t_j} T_j \\
 -\sqrt{t_j} T_j  & X - t_j T_j\\
\end{array}
\right] \geq 0 \right\} _{j = 1}^{m} 
\end{array}
\right\}.  
\end{equation}
In the above expression, $\{w_j\}_{j=1}^m$ and $\{t_j\}_{j=1}^m$ are the weights and nodes of the $m$-point Gauss-Legendre quadrature on the interval $[0,1]$, and $k$ is an integer that is picked to achieve the desired accuracy (see App.~\ref{sec:P_ln_approx} for more details). The semidefinite program in~\eqref{eq:meas_ch_RE_SDP} was also presented in~\cite{huang2024semi} along with a proof. We still include a proof sketch of~\eqref{eq:meas_ch_RE_SDP} in App.~\ref{sec:P_ln_approx} for the convenience of the reader. 

\subsection{Approximating the perspective of the operator logarithm \texorpdfstring{$P_{\ln}$}{P ln}}

\label{sec:P_ln_approx}

In this section, we present a semidefinite approximation of the set $\left\{\Theta_{AB}: \Theta_{AB}\le -D_{\operatorname{op}}\!\left(\rho_A\otimes I_B, \Omega_{AB}\right) \right\}$, which is the only missing piece required to go from the variational expression in~\eqref{eq:G7} to the semidefinite program in~\eqref{eq:meas_ch_RE_SDP}. The semidefinite approximation of this set was discovered in~\cite{fawzi2019semidefinite}, and we include a concise version of the proof here for completeness.

\noindent
We begin with the following integral representation of the logarithm function:
\begin{align}
\ln(z) = \int_0^1 f_t(z) \, dt, \qquad
f_t(z) \coloneqq \frac{z-1}{t(z-1)+1}.
\end{align}
Note that $f_t(z)$ has the following semidefinite representation:
\begin{align}
\label{eq:f_t_scalar}
f_t(z)\geq \tau \qquad \iff \qquad
\left[
\begin{array}{cc}
 z-1-\tau & -\sqrt{t} \tau \\
 -\sqrt{t} \tau & 1-t\tau \\
\end{array}
\right] \geq 0,
\end{align}
which can be verified using the Schur complement lemma.

Let $\{w_j\}_{j=1}^m$ and $\{t_j\}_{j=1}^m$ be the weights and nodes of the $m$-point Gauss--Legendre quadrature on the interval $[0,1]$. Then,
\begin{align}
\ln(z) \approx \sum_{j=1}^m w_j f_{t_j}(z) \eqqcolon r_m(z),
\end{align}
where the difference between $\ln(z)$ and its approximation $r_m(z)$ decreases exponentially with increasing $m$ (see~\cite{fawzi2019semidefinite} for details). Importantly, $r_m(z)$ has a semidefinite representation.

The approximation is tighter when $z$ is close to one. Now consider the following function:
\begin{align}
    r_{m,k}(z) &\coloneqq 2^kr_m\!\left(z^{1/2^k}\right)\\
    &\approx 2^k\ln\!\left(z^{1/2^k}\right)\\
    &=\ln(z).
\end{align}
Since $z^{1/2^k}$ is closer to one than $z$, we have $\left|r_m\!\left(z^{1/2^k}\right) - \ln\!\left(z^{1/2^k}\right)\right| \le \left|r_m(z) - \ln(z)\right|$, making $r_{m,k}(z)$ a better approximation of $\ln(z)$ than $r_m(z)$.

Applying the same analysis to a positive definite operator $Z$, we have that
\begin{align}
\ln (Z) \approx 2^k\sum_{j = 1} ^m w_j f_{t_j}(Z^{2^{-k}}),
\end{align}
and
\begin{equation}\label{eq:ln_SDP_rep_op}
    f_t(Z) \ge \widetilde{T} \iff \left[
    \begin{array}{cc}
        Z - \openone - \widetilde{T} & -\sqrt{t}\widetilde{T} \\
         -\sqrt{t}\widetilde{T}& \openone - t\widetilde{T} 
    \end{array}
    \right] \ge 0 \qquad \forall t\in [0,1], \ \widetilde{T}\in \operatorname{Herm}.
\end{equation}

We are interested in the operator relative entropy of two positive semidefinite operators, and since the operator relative entropy has desirable continuity properties, we will restrict our discussion to positive definite operators \cite{fujii1989relative}.

The operator relative entropy can be approximated as follows:
\begin{equation}\label{eq:op_rel_ent_sdp_fn}
    D_{\operatorname{op}}(X\Vert Y) = -X^{\frac{1}{2}}\ln\!\left(X^{-\frac{1}{2}}YX^{-\frac{1}{2}}\right)X^{\frac{1}{2}}\approx -2^k\sum_{j=1}^m w_jX^{\frac{1}{2}}f_{t_j}\!\left(\left(X^{-\frac{1}{2}}YX^{-\frac{1}{2}}\right)^{2^{-k}}\right)X^{\frac{1}{2}}.
\end{equation}

Note that
\begin{equation}
    \left(X^{-\frac{1}{2}}YX^{-\frac{1}{2}}\right)^{2^{-k}} = G_{2^{-k}}\!\left(\openone,X^{-\frac{1}{2}}YX^{-\frac{1}{2}}\right),
\end{equation}
where $G_t(\cdot,\cdot)$ is the weighted operator geometric mean defined in~\eqref{eq:geo_mean_defn}. The weighted operator geometric mean of two positive definite operators has a semidefinite representation, which implies that for all $k+1$ Hermitian operators $\widetilde{Z}_0,\widetilde{Z}_1,\ldots, \widetilde{Z}_k$:
\begin{equation}
    \left\{\left[\begin{array}{cc}
       \openone  & \widetilde{Z}_{i+1} \\
        \widetilde{Z}_{i+1} & \widetilde{Z}_i
    \end{array}\right]\ge 0\right\}_{i=0}^k \land \widetilde{Z}_0 = X^{-\frac{1}{2}}YX^{-\frac{1}{2}} \iff \widetilde{Z}_j \le G_{2^{-j}}\!\left(\openone,X^{-\frac{1}{2}}YX^{-\frac{1}{2}}\right) \qquad \forall j\in \left\{0,1,\ldots,k\right\}.
\end{equation}
The above equation is not in a semidefinite representable form, since we have the term $X^{-\frac{1}{2}}YX^{-\frac{1}{2}}$. Let us define $Z_i \coloneqq X^{\frac{1}{2}}\widetilde{Z}_iX^{\frac{1}{2}}$. Since $X$ is positive definite, we can write
\begin{equation}
    \left\{\left[\begin{array}{cc}
       X  & Z_{i+1} \\
        Z_{i+1} & Z_i
    \end{array}\right]\ge 0\right\}_{i=0}^k \land Z_0 = Y \iff Z_j \le X^{\frac{1}{2}}G_{2^{-j}}\!\left(\openone,X^{-\frac{1}{2}}YX^{-\frac{1}{2}}\right)X^{\frac{1}{2}} \qquad \forall j\in \left\{0,1,\ldots,k\right\},
\end{equation}
which shows that $X^{\frac{1}{2}}\left(X^{-\frac{1}{2}}YX^{-\frac{1}{2}}\right)^{2^{-k}}X^{\frac{1}{2}}$ has a semidefinite representation.

Now, let us go back to the function $f_t(\cdot)$. Note that $\widetilde{Z}_k = \left(X^{-\frac{1}{2}}YX^{-\frac{1}{2}}\right)^{2^{-k}}$. From~\eqref{eq:ln_SDP_rep_op} we know that for every $t\in [0,1]$ and every Hermitian operator $\widetilde{T}$
\begin{align}
      f_t\!\left(\widetilde{Z}_k\right) \ge \widetilde{T} &\iff \left[\begin{array}{cc}
        \widetilde{Z}_k - \openone - \widetilde{T} & -\sqrt{t}\widetilde{T} \\
         -\sqrt{t}\widetilde{T}& \openone - t\widetilde{T} 
    \end{array}
    \right] \ge 0\\
    &\iff \left[
    \begin{array}{cc}
        X^{\frac{1}{2}} & 0 \\
         0 & X^{\frac{1}{2}}
    \end{array}
    \right]\left[\begin{array}{cc}
        \widetilde{Z}_k - \openone - \widetilde{T} & -\sqrt{t}\widetilde{T} \\
         -\sqrt{t}\widetilde{T}& \openone - t\widetilde{T} 
    \end{array}
    \right]\left[
    \begin{array}{cc}
        X^{\frac{1}{2}} & 0 \\
         0 & X^{\frac{1}{2}}
    \end{array}
    \right] \ge 0\\
    &\iff \left[
    \begin{array}{cc}
        X^{\frac{1}{2}}\widetilde{Z}_kX^{\frac{1}{2}} - X - X^{\frac{1}{2}}\widetilde{T}X^{\frac{1}{2}} & -\sqrt{t}X^{\frac{1}{2}}\widetilde{T}X^{\frac{1}{2}} \\
         -\sqrt{t}X^{\frac{1}{2}}\widetilde{T}X^{\frac{1}{2}} & X -tX^{\frac{1}{2}}\widetilde{T}X^{\frac{1}{2}}
    \end{array}
    \right] \ge 0. \label{eq:ln_SDP_geo_mean_SDP_cl}
\end{align}
     Since $X^{\frac{1}{2}}\left(X^{-\frac{1}{2}}YX^{-\frac{1}{2}}\right)^{2^{-k}}X^{\frac{1}{2}}$ has a semidefinite representation, we conclude that $f_{t_j}\!\left(\left(X^{-\frac{1}{2}}YX^{-\frac{1}{2}}\right)^{2^{-k}}\right)$ also has a semidefinite representation. Replacing $\widetilde{T}$ with $X^{\frac{1}{2}}\widetilde{T}X^{\frac{1}{2}}$ in~\eqref{eq:ln_SDP_geo_mean_SDP_cl}, we can write 
    \begin{align}
        X^{\frac{1}{2}}f_t(\widetilde{Z}_k)X^{\frac{1}{2}} \ge T &\iff f_t(\widetilde{Z}_k) \ge X^{-\frac{1}{2}}TX^{-\frac{1}{2}}\\
        &\iff \left[
    \begin{array}{cc}
        X^{\frac{1}{2}}\widetilde{Z}_kX^{\frac{1}{2}} - X - T & -\sqrt{t}T \\
         -\sqrt{t}T & X -tT
    \end{array}
    \right] \ge 0\\
    &\iff \left[\begin{array}{cc}
        Z_k - X - T & -\sqrt{t}T \\
         -\sqrt{t}T & X -tT
    \end{array}\right]\ge 0,
    \end{align}
    where the final implication follows from the fact that $Z_k = X^{\frac{1}{2}}\widetilde{Z}_kX^{\frac{1}{2}}$.
    Put together, $X^{\frac{1}{2}}f_{t}\!\left(\left(X^{-\frac{1}{2}}YX^{-\frac{1}{2}}\right)^{2^{-k}}\right)X^{\frac{1}{2}}\ge T$ if and only if
    \begin{equation}
        \left\{\left[\begin{array}{cc}
       X  & Z_{i+1} \\
        Z_{i+1} & Z_i
    \end{array}\right]\ge 0\right\}_{i=0}^k, \ Z_0 = Y,\ \left[\begin{array}{cc}
        Z_k - X - T & -\sqrt{t}T \\
         -\sqrt{t}T & X -tT
    \end{array}\right] \ge 0.
    \end{equation}
Consequently,
\begin{equation}
    \sum_{j=1}^m w_jX^{\frac{1}{2}}f_{t_j}\!\left(\left(X^{-\frac{1}{2}}YX^{-\frac{1}{2}}\right)^{2^{-k}}\right)X^{\frac{1}{2}} \ge \sum_{j=1}^m w_jT_j
\end{equation}
if and only if
\begin{equation}
        \left\{\left[\begin{array}{cc}
       X  & Z_{i+1} \\
        Z_{i+1} & Z_i
    \end{array}\right]\ge 0\right\}_{i=0}^k,\  Z_0 = Y,\ \left\{\left[\begin{array}{cc}
        Z_k - X - T_j & -\sqrt{t_j}T_j \\
         -\sqrt{t_j}T_j & X -t_jT_j
    \end{array}\right] \ge 0\right\}_{j=1}^m.
    \end{equation}
Comparing with~\eqref{eq:op_rel_ent_sdp_fn}, we conclude that the set $\left\{\Theta: \Theta \le -D_{\operatorname{op}}(X\Vert Y)\right\}$ can be approximated by the following set described by semidefinite constraints:
\begin{equation}
    \label{eq:K_cone}
    \left\{\begin{array}{c}
         \Theta:  \\
         \sum_{j=1}^m w_jT_j = 2^{-k}\Theta, \quad Z_0 = Y,\\
         \left[\begin{array}{cc}
       X  & Z_{i+1} \\
        Z_{i+1} & Z_i
    \end{array}\right]\ge 0 \qquad \forall i\in \{0,1,\ldots,k\},\\
    \left[\begin{array}{cc}
        Z_k - X - T_j & -\sqrt{t_j}T_j \\
         -\sqrt{t_j}T_j & X -t_jT_j
    \end{array}\right] \ge 0 \qquad \forall j\in \{1,2,\ldots,m\}
    \end{array}\right\}.
\end{equation}

\section{Proof of Equation~\eqref{eq:type_II_err_ub_bel_stas}}

\label{app:up-bnd-ecbs-proof}

In this section, we prove that the energy-constrained Belavkin--Staszewski relative entropy between two channels is an upper bound on the error exponent of asymmetric energy-constrained channel discrimination, as stated in~\eqref{eq:type_II_err_ub_bel_stas}.

\begin{proposition}
For an arbitrary energy-constrained channel discrimination protocol of the form described in Section~\ref{sec:q-ch-disc}, for channels $\mathcal{N}$ and $\mathcal{M}$, for a grounded Hamiltonian $H$, and for all $\varepsilon\in [0,1)$, $E>0$, and $n\in \mathbb{N}$, the following bound holds:
\begin{equation}
    \frac{-\ln \beta^\star_n(E,\varepsilon)}{n} \leq \frac{1}{1-\varepsilon}\left(\widehat{D}_{H,E}\!\left(\mathcal{N}_{A\to B}\middle\Vert\mathcal{M}_{A\to B}\right) + \frac{h_2(\varepsilon)}{n}\right),
    \label{eq:non-asymptotic-bnd-ecbs}
\end{equation}
where the energy-constrained Belavkin--Staszewski channel relative entropy $\widehat{D}_{H,E}\!\left(\mathcal{N}_{A\to B}\middle\Vert\mathcal{M}_{A\to B}\right)$ is defined in~\eqref{eq:def-ecbs} and 
\begin{equation}
h_2(\varepsilon)\coloneqq -\varepsilon\ln(\varepsilon) - (1-\varepsilon)\ln(1-\varepsilon)    
\end{equation}
is the binary entropy. 
\end{proposition}

\begin{proof}
Let us denote the probe state for the $j^{\operatorname{th}}$ channel by $\tau^{(j,\mathcal{E})}_{RA}$, which can be mathematically expressed as follows, as is evident from Fig.~\ref{fig:scheme}:
\begin{equation}\label{eq:ad_ch_disc_probe_defn}
    \tau^{(j,\mathcal{E})}_{RA} \coloneqq \left(\mathcal{A}^{(j-1)}_{RB\to RA}\circ\mathcal{E}_{A\to B}\right)\left(\tau^{(j-1,\mathcal{E})}_{RA}\right) \qquad \forall j\in \{2,3,\ldots,n\},
\end{equation}
with $\tau^{(1,\mathcal{E})}_{RA} \coloneqq \tau_{RA}$ being the initial probe state. $\mathcal{E}_{A\to B}$ represents the unknown channel. As such, after $n$ uses of the channel, the final state is either $\rho^{(n)}_{\mathcal{N}} \coloneqq \mathcal{N}_{A\to B}\!\left(\tau^{(n,\mathcal{N})}_{RA}\right)$ or $\rho^{(n)}_{\mathcal{M}}  \coloneqq \mathcal{M}_{A\to B}\!\left(\tau^{(n,\mathcal{M})}_{RA}\right)$.

The final measurement, denoted by $Q$ in Fig.~\ref{fig:scheme}, attempts to distinguish between the states $\rho^{(n)}_{\mathcal{N}}$ and $\rho^{(n)}_{\mathcal{M}}$. As such, the task of distinguishing between channels $\mathcal{N}$ and $\mathcal{M}$, with some fixed adaptive channels $\left\{\mathcal{A}^{(1)}_{RB\to RA}, \mathcal{A}^{(2)}_{RB\to RA},\ldots, \mathcal{A}^{(n-1)}_{RB\to RA}\right\}$, is equivalent to the task of distinguishing between states $\rho^{(n)}_{\mathcal{N}}$ and $\rho^{(n)}_{\mathcal{M}}$. If the type-I error probability is required to be  $\leq \varepsilon$, the minimum type-II error probability in state discrimination can be expressed in terms of the hypothesis testing relative entropy~\cite{Wang2012}. In particular, let $\alpha_n$ and $\beta_n$ denote the type-I and type-II error probabilities, respectively; then $\alpha_n \le \varepsilon$ implies that
\begin{equation}
    -\ln \beta^\star_n(E,\varepsilon) \le D^{\varepsilon}_h\!\left(\rho^{(n)}_{\mathcal{N}}\middle\Vert\rho^{(n)}_{\mathcal{M}}\right),
\end{equation}
where $D^{\varepsilon}_h(\cdot\Vert\cdot)$ denotes the hypothesis-testing relative entropy. It is well known that the following inequality holds for any two states $\rho$ and $\sigma$ \cite[Eq.~(2)]{Wang2012}:
\begin{align}
    D^{\varepsilon}_h\!\left(\rho\middle\Vert\sigma\right) &\le \frac{1}{1-\varepsilon}\!\left(D\!\left(\rho\middle\Vert\sigma\right) + h_2(\varepsilon)\right) \qquad \forall \varepsilon \in [0,1),\\
    &\le \frac{1}{1-\varepsilon}\!\left(\widehat{D}\!\left(\rho\middle\Vert\sigma\right) + h_2(\varepsilon)\right) \qquad \forall \varepsilon \in [0,1),
\end{align}
where the last inequality follows from the fact that the Belavkin--Staszewski relative entropy is larger than or equal to the Umegaki relative entropy~\cite{HiaiPetz1993}. As such,
\begin{align}
    -\ln \beta^\star_n(E,\varepsilon) &\le \frac{1}{1-\varepsilon}\!\left(\widehat{D}\!\left(\rho^{(n)}_{\mathcal{N}}\middle\Vert\rho^{(n)}_{\mathcal{M}}\right) + h_2(\varepsilon)\right)\\
    \implies \frac{-\ln \beta^\star_n(E,\varepsilon)}{n} &\le \frac{1}{n(1-\varepsilon)}\!\left(\widehat{D}\!\left(\rho^{(n)}_{\mathcal{N}}\middle\Vert\rho^{(n)}_{\mathcal{M}}\right) + h_2(\varepsilon)\right)\label{eq:tII_err_exp_le_geo}.
\end{align}
Let us now simplify the Belavkin--Staszewski divergence on the right-hand side of~\eqref{eq:tII_err_exp_le_geo}. Recall the inequality in~\eqref{eq:BS_decomp_ch_div_st_div}. We can write
\begin{align}
    \widehat{D}\!\left(\rho^{(n)}_{\mathcal{N}}\middle\Vert\rho^{(n)}_{\mathcal{M}}\right) &= \widehat{D}\!\left(\mathcal{N}_{A\to B}\!\left(\tau^{(n,\mathcal{N})}_{RA}\right)\middle\Vert\mathcal{M}_{A\to B}\!\left(\tau^{(n,\mathcal{M})}_{RA}\right)\right)\label{eq:rho_n_to_tau_n}\\
    &\le \widehat{D}\!\left(\tau^{(n,\mathcal{N})}_{RA}\middle\Vert\tau^{(n,\mathcal{M})}_{RA}\right) + \operatorname{Tr}\!\left[\left(\tr_{R}\!\left[\tau^{(n,\mathcal{N})}_{RA}\right]\right)^T D_{\operatorname{op}}\!\left(\Gamma^{\mathcal{N}}_{AB}\middle\|\Gamma^{\mathcal{M}}_{AB}\right)\right],
\end{align}
where the transpose is with respect to the eigenbasis of the Hamiltonian. Now recall the definition of $\tau^{n,\mathcal{E}}_{RA}$ from~\eqref{eq:ad_ch_disc_probe_defn}. As such,
\begin{align}
    \widehat{D}\!\left(\tau^{(n,\mathcal{N})}_{RA}\middle\Vert\tau^{(n,\mathcal{M})}_{RA}\right) &= \widehat{D}\!\left(\mathcal{A}^{(n-1)}_{RB\to RA}\circ\mathcal{N}_{A\to B}\!\left(\tau^{(n-1,\mathcal{N})}_{RA}\right)\middle\Vert\mathcal{A}^{(n-1)}_{RB\to RA}\circ\mathcal{M}_{A\to B}\!\left(\tau^{(n-1,\mathcal{M})}_{RA}\right)\right)\\
    &\le \widehat{D}\!\left(\mathcal{N}_{A\to B}\!\left(\tau^{(n-1,\mathcal{N})}_{RA}\right)\middle\Vert\mathcal{M}_{A\to B}\!\left(\tau^{(n-1,\mathcal{M})}_{RA}\right)\right),\label{eq:tau_n_to_tau_n-1}
\end{align}
where the inequality follows from the data-processing inequality for the Belavkin--Staszewski relative entropy (see, e.g., \cite[Cor.~7.55]{khatri2020principles}). Note that the expression on the right-hand side of~\eqref{eq:tau_n_to_tau_n-1} has the same form as the expression in~\eqref{eq:rho_n_to_tau_n}. Therefore,  the following upper bound on $\widehat{D}\!\left(\rho^{(n)}_{\mathcal{N}}\middle\Vert\rho^{(n)}_{\mathcal{M}}\right)$ is a consequence of a recursive argument:
\begin{align}
    \widehat{D}\!\left(\rho^{(n)}_{\mathcal{N}}\middle\Vert\rho^{(n)}_{\mathcal{M}}\right) &\le \widehat{D}\!\left(\tau^{(1,\mathcal{N})}_{RA}\middle\Vert\tau^{(1,\mathcal{M})}_{RA}\right) + \sum_{j = 1}^n \operatorname{Tr}\!\left[\left(\tr_{R}\!\left[\tau^{(j,\mathcal{N})}_{RA}\right]\right)^T D_{\operatorname{op}}\!\left(\Gamma^{\mathcal{N}}_{AB}\middle \|\Gamma^{\mathcal{M}}_{AB}\right)\right]\\
    &= n \operatorname{Tr}\!\left[\left(\frac{1}{n}\sum_{j = 1}^n\left(\tr_{R}\!\left[\tau^{(j,\mathcal{N})}_{RA}\right]\right)^T\right) D_{\operatorname{op}}\!\left(J^{\mathcal{N}}_{AB}\middle \| J^{\mathcal{M}}_{AB}\right)\right].\label{eq:BS_fin_states_to_ore}
\end{align}
Note that $\frac{1}{n}\sum_{j = 1}^n\left(\tr_{R}\!\left[\tau^{(j,\mathcal{N})}_{RA}\right]\right)^T$ is a quantum state. Irrespective of what adaptive strategy is chosen, the average energy constraint ensures that 
\begin{equation}
    \operatorname{Tr}\!\left[H_A\!\left(\frac{1}{n}\sum_{j = 1}^n\left(\tr_{R}\!\left[\tau^{(j,\mathcal{N})}_{RA}\right]\right)^T\right)\right] = \frac{1}{n}\sum_{j=1}^n\operatorname{Tr}\!\left[H_A\operatorname{Tr}_R\!\left[\tau^{(j,\mathcal{N})}_{RA}\right]\right] \le E,
\end{equation}
where the equality follows from the linearity of the trace and the fact that the transpose is taken with respect to the eigenbasis of the Hamiltonian. The inequality in~\eqref{eq:BS_fin_states_to_ore} can then be bounded from above as follows:
\begin{align}
    \widehat{D}\!\left(\rho^{(n)}_{\mathcal{N}}\middle\Vert\rho^{(n)}_{\mathcal{M}}\right) &\le n\sup_{\substack{\omega\ge 0, \tr[\omega] = 1,\\ \tr\left[H_A\omega_A\right] \le E}}  \tr\!\left[\omega_A D_{\operatorname{op}}\!\left(J^{\mathcal{N}}_{AB}\middle \| J^{\mathcal{M}}_{AB}\right)\right]\\
    &= n\widehat{D}_{H,E}\!\left(\mathcal{N}_{A\to B}\middle\Vert\mathcal{M}_{A\to B}\right),
\end{align}
where the equality follows from Lemma~\ref{lem:ECBS_op_rel_ent_form}. Substituting the above inequality into~\eqref{eq:tII_err_exp_le_geo}, we arrive at the claim in~\eqref{eq:non-asymptotic-bnd-ecbs}.
\end{proof}

Taking the limit of~\eqref{eq:non-asymptotic-bnd-ecbs} as $n\to \infty$ and $\varepsilon \to 0$, we then conclude~\eqref{eq:type_II_err_ub_bel_stas}. That is,
\begin{align}
    \lim_{\varepsilon \to 0}\limsup_{n \to \infty} \frac{-\ln \beta^\star_n(E,\varepsilon)}{n} & \leq \lim_{\varepsilon \to 0}\limsup_{n \to \infty} \frac{1}{1-\varepsilon}\left(\widehat{D}_{H,E}\!\left(\mathcal{N}_{A\to B}\middle\Vert\mathcal{M}_{A\to B}\right) + \frac{h_2(\varepsilon)}{n}\right) \\
    & = \lim_{\varepsilon \to 0}\frac{1}{1-\varepsilon}\widehat{D}_{H,E}\!\left(\mathcal{N}_{A\to B}\middle\Vert\mathcal{M}_{A\to B}\right) \\
    & = \widehat{D}_{H,E}\!\left(\mathcal{N}_{A\to B}\middle\Vert\mathcal{M}_{A\to B}\right) .
\end{align}

\section{Bosonic dephasing and loss channels and their Choi matrices}

\label{sec:CHOI}

\subsection{Bosonic dephasing channel}

A bosonic dephasing channel (BDC) describes a process in which the state of the system is phase-shifted by a random phase $\phi\in [-\pi,\pi]$ drawn with respect to some probability density $p(\phi)$ on this interval. The action of a BDC on an arbitrary state $\rho$ can be described mathematically as follows:
\begin{equation}
    \mathcal{D}_p(\rho) = \int_{-\pi}^{\pi}p(\phi)e^{-i\hat{n}\phi}\rho e^{i\hat{n}\phi},
\end{equation}
where $\hat{n}$ is the photon number operator.

An important class of BDCs emerges when the underlying probability density is set to be the wrapped normal distribution~\cite{Jiang2010,Arqand2020} (see also~\cite{lami2023exact,huang2024exact}):
\begin{equation}
    p_{\gamma}(\phi) \coloneqq \frac{1}{\sqrt{2\pi\gamma}}\sum_{k=-\infty}^{\infty} e^{-\frac{1}{2\gamma}(\phi + 2\pi k)^2},
\end{equation}
where $\gamma>0$ is the variance of the probability density. The action of such a dephasing channel, say $\mathcal{D}_{\gamma}$, on an arbitrary state $\rho$ takes the following form (see \cite[Eqs.~(2), (3), (11)]{lami2023exact}:
\begin{equation}\label{eq:dephasing_on_rho}
    \mathcal{D}_{\gamma}(\rho) = \sum_{m,n=0}^{\infty}e^{-\frac{\gamma}{2}(m-n)^2}\langle m|\rho|n\rangle ~|m\rangle\!\langle n|,
\end{equation}
where $\left\{|i\rangle\right\}_{i=0}^{\infty}$ is the Fock basis.

From~\eqref{eq:Choi_op_defn} and~\eqref{eq:dephasing_on_rho}, it is straightforward to verify that the Choi operator of $\mathcal{D}_{\gamma}$ is equal to the following:
\begin{equation}
    J^{\mathcal{D}_{\gamma}}_{RB} = \sum_{m,n=0}^{\infty} |m\rangle\!\langle n|_R \otimes e^{-\frac{\gamma}{2}(m-n)^2}|m\rangle\!\langle n|_B.
\end{equation}

\subsection{Pure-loss channel}

The loss of photons in a quantum information processing protocol is mathematically modeled by the pure-loss channel, which has been considered extensively in quantum information theory (see, e.g.,~\cite{Giovannetti2004,Wolf2007,Nair2018}).

Let $\hat{a}$ denote the annihilation operator for the system $A$ of interest, and let $\hat{b}$ denote the annihilation operator for the system that is lost to the environment $E$. Consider the beam-splitter unitary operator with transmissivity $\eta\in [0,1]$:
\begin{equation}
    U_{\operatorname{BS}}(\eta) \coloneqq e^{\cos^{-1}(\sqrt{\eta})(a^{\dagger}b - b^{\dagger}a)}.
\end{equation}
The action of $U_{\operatorname{BS}}(\eta)$ on the creation operator $\hat{a}^{\dagger}$ is given as follows:
\begin{equation}\label{eq:BS_unitary_on_creation}
    U_{\operatorname{BS}}(\eta)\hat{a}^{\dagger} U_{\operatorname{BS}}(\eta)^{\dagger} = \sqrt{\eta}~\hat{a}^{\dagger} + \sqrt{1-\eta}~\hat{b}^{\dagger}.
\end{equation}
The action of a pure-loss channel, with transmissivity $\eta$, on an arbitrary state $\rho_A$ is as follows:
\begin{equation}\label{eq:loss_ch_defn-app}
    \mathcal{L}_{\eta}\!\left(\rho_A\right) \coloneqq \operatorname{Tr}_E\!\left[U_{\operatorname{BS}}(\eta)\!\left(\rho_A\otimes |0\rangle\!\langle 0|_E\right)U_{\operatorname{BS}}(\eta)^{\dagger}\right].
\end{equation}

To determine the Choi operator of the pure-loss channel, let us note that
\begin{equation}\label{eq:BS_unitary_on_vaccuum}
    U_{\operatorname{BS}}(\eta)|0\rangle_A|0\rangle_E = |0\rangle_A|0\rangle_E \qquad \forall \eta \in [0,1].
\end{equation}
Now consider the following state vector:
\begin{align}
U_{\operatorname{BS}}(\eta)\ket{n}_A|0\rangle_E &= U_{\operatorname{BS}}(\eta)\frac{1}{\sqrt {n !}} \hat a^{\dag n} \ket{0}_A\ket{0}_E\\
&= \frac{\left(U_{\operatorname{BS}}(\eta)a^{\dagger}U_{\operatorname{BS}}(\eta)^{\dagger}\right)^n}{\sqrt{n!}}U_{\operatorname{BS}}(\eta)|0\rangle_A|0\rangle_E\\
&= \frac{\left(U_{\operatorname{BS}}(\eta)a^{\dagger}U_{\operatorname{BS}}(\eta)^{\dagger}\right)^n}{\sqrt{n!}}|0\rangle_A|0\rangle_E\\
&= \frac{1}{\sqrt{n!}}  \sum_{k=0}^n \binom{n}{k}  \sqrt{ \eta ^{n-k} (1-\eta)^{k} }\hat a ^{\dag (n-k)} \hat b^{\dag k}\ket{0}_A\ket{0}_E  \\
 &=  \sum_{k=0}^n \binom{n}{k}  \sqrt{ \eta ^{n-k} (1-\eta)^{k} } \sqrt{ \frac{(n-k)!k!}{n!}} \ket{n-k}_A\ket{k}_E \\
 &=\sum_{k=0}^n   \sqrt{ \binom{n}{k} \eta ^{n-k} (1-\eta)^{k} }  \ket{n-k}_A\ket{k}_E,\label{eq:BS_unitary_on_Fock}
\end{align}
where the first equality follows from the creation operator representation of Fock states, the second equality follows from the unitarity of $U_{\operatorname{BS}}(\eta)$, the third equality follows from~\eqref{eq:BS_unitary_on_vaccuum}, and the fourth equality follows from~\eqref{eq:BS_unitary_on_creation}. Now consider the following equality:
\begin{multline}\label{eq:BS_unitary_Fock_op}
U_{\operatorname{BS}}(\eta)\!\left(\ket{m}\!\bra{n}_A \otimes \ket{0}\!\bra{0}_E\right)U_{\operatorname{BS}}(\eta)^{\dagger} \\=
\sum_{k_1=0}^{m}\sum_{k_2=0}^{n}
\sqrt{\binom{n_1}{k_1}\binom{n_2}{k_2}\eta^{m+n-k_1-k_2}(1-\eta)^{k_1+k_2}} \ket{m-k_1}\!\bra{n-k_2}_A
\otimes\ket{k_1}\!\bra{k_2}_E,
\end{multline}
which follows from~\eqref{eq:BS_unitary_on_Fock}. As such,
\begin{equation}
\mathcal{L}_{\eta}\!\left(|m\rangle\!\langle n|\right) = \sum_{k=0}^{\min(m,n)}
\sqrt{\binom{m}{k_1}\binom{n}{ k_2}\eta^{m+n-2k}(1-\eta)^{2k}} \ket{m-k}\!\bra{n-k},
\end{equation}
where we have used~\eqref{eq:loss_ch_defn-app} and~\eqref{eq:BS_unitary_Fock_op} to arrive at the above equality. The Choi operator of the pure-loss channel, with transmissivity $\eta$, can hence be written as follows:
\begin{equation}
    J^{\mathcal{L}_{\eta}}_{RB} = \sum_{m,n=0}^{\infty}\sum_{k=0}^{\min(m,n)}
\sqrt{\binom{m} { k_1}\binom{n} { k_2}}\eta^{\frac{1}{2}(m+n)-k}(1-\eta)^{k} |m\rangle\!\langle n|_R\otimes\ket{m-k}\!\bra{n-k}_B.
\end{equation}

\subsection{Loss-dephasing channel}

If a probe state experiences both loss and dephasing, the corresponding channel is called the loss-dephasing channel, and its Choi matrix is then given by
\begin{equation}
    J^{\eta,\gamma}_{RB} = \sum_{m,n=0}^{\infty}\sum_{k=0}^{\min(m,n)}
\sqrt{\binom{m} { k_1}\binom{n} { k_2}}\eta^{\frac{1}{2}(m+n)-k}(1-\eta)^{k} e^{-\frac{\gamma}{2}(m-n)^2} |m\rangle\!\langle n|_R\otimes\ket{m-k}\!\bra{n-k}_B.
\end{equation}

\section{Channel divergences and their SDP representations}

\label{app:GRD-SDP}

In this appendix, we present the semidefinite programs that we use to compute the channel divergences relevant to this work. 

In general, it is not feasible to compute these channel divergences between bosonic channels since they do not have a finite-dimensional Choi operator. To facilitate the computation of these quantities, we restrict the dimension of the probe state in the channel discrimination task.

\subsection{Geometric R\'enyi divergence}

Recall from~\eqref{eq:EC_geo_ch_final} that the $\alpha$-geometric R\'enyi relative entropy of channels can be expressed as an optimization over all states with energy less than or equal to the energy budget $E$. Since we fix the dimension of the probe state to be a finite integer, the optimization in~\eqref{eq:EC_geo_ch_final} is a semidefinite program. 
The state divergence for two states is defined as~\cite{matsumoto2015new}
\begin{align}
\widehat{D}_{\alpha}(\rho\|\sigma) \coloneqq  \frac{1}{\alpha-1} \ln 
 \tr\!\left[ \sigma^{\frac{1}{2}}\left(\sigma^{-\frac{1}{2}}\rho  \sigma^{-\frac{1}{2}} \right)^\alpha      \sigma^{\frac{1}{2}} \right]
 \qquad 
 \alpha \in (1,2]
\end{align}
The quantity $\widehat D(\rho\|\sigma)$ is known to be the maximal divergence among all quantum
R\'enyi divergences. In the limit that $\alpha\rightarrow 1$ converges to the Belavkin-Staszewski relative entropy. The geometric R\'enyi divergence for two channels $\curlN$ and $\curlM$ is defined as
\begin{align}
\widehat D(\curlN\|\curlM) \coloneqq  \max_{\phi_{RA}}
\widehat D(\curlN_{A\rightarrow B}(\phi_{RA}) \|\curlM_{A\rightarrow B}(\phi_{RA}))
\end{align}

It has a closed form:
\begin{align} 
\label{eq:close_form}
    \widehat D_\alpha( \curlN_{A\rightarrow B}  \| \curlM_{A\rightarrow B}) 
      &= \frac{1}{\alpha -1}  \ln 
    \left\| \tr_B \!\left[  J_\curlM^{\frac{1}{2}}  \left(  J_\curlM^{-\frac{1}{2}}    J_\curlN      J_\curlM^{-\frac{1}{2}} \right )^\alpha J_\curlM^{\frac{1}{2}} \right]  \right\|_\infty \\
    &\equiv \frac{1}{\alpha - 1} \ln  \left\|  \tr_B \!\left[G_{1-\alpha}(J_{RB}^\curlN,J_{RB}^M)\right] \right\|_\infty\label{eq:ch_geo_entropy_geo_mean_op_norm_form}
\end{align}

Lemma~3 of~\cite{FangFawzi2021} shows that~\eqref{eq:close_form} for $\alpha(\ell) = 1+ 2^{-\ell}, ~\ell \in \mathbb{N}$, (without an energy constraint on $\rho_A$) can be written as the following semidefinite program,
\begin{align} \label{eq:lemma3fang}
&2^\ell ~ \ln  \min y \quad \text{s.t} \quad \left[L_{RB},~ \{N_i\}_{i = 0}^\ell, ~ J_\curlM,y\right]_H,\notag\\
\left[
\begin{array}{cc}
L_{RB} & J^\curlN \\
 J^\curlN & N_\ell \\
\end{array}
\right]_{ P}, \qquad
&\left\{
\left[
\begin{array}{cc}
J_\curlN & N_i \\
 N_i & N_{i-1} \\
\end{array}
\right]_P
\right\}_{i = 1}^\ell, \qquad
\left[N_0-J_\curlM\right]_E, \qquad
\left[y \openone_A -\text{Tr}_B\!\left[ L_{RB}\right]\right]_P.
\end{align}
The notation here $\left[M\right]_P,\left[M\right]_E, \left[M\right]_H$ denote the positive semidefinite condition $M\geq 0$, the equality condition $M=0$ and Hermiticity, respectively. The proofs are given in Ref.~\cite{fawzi2019semidefinite}, and we include them in \mmw{REF} App.~\ref{app:w_geo_sdp} for completeness.

\subsection{Proof of Equation~\eqref{eq:lemma3fang}}

\label{app:w_geo_sdp}

Let $X$ and $Y$ be two positive semidefinite operators such that $\operatorname{supp}(X)\subseteq \operatorname{supp}(Y)$. We know from the Schur complement lemma that for every Hermitian operator $N$,
\begin{equation}
    \begin{pmatrix}
    X & N\\
    N & Y
    \end{pmatrix} \ge 0 \iff Y \ge NX^{-1}N, \quad (\openone-XX^{-1})N = 0,
\end{equation}
where the inverse is taken on the support of $X$. Now consider the following inequalities:
\begin{align}
    Y &\ge NX^{-1}N\\
    \implies X^{-\frac{1}{2}}YX^{-\frac{1}{2}} &\ge X^{-\frac{1}{2}}NX^{-1}NX^{-\frac{1}{2}}\\
    \implies X^{-\frac{1}{2}}YX^{-\frac{1}{2}} &\ge \left(X^{-\frac{1}{2}}NX^{-\frac{1}{2}}\right)^2\\
    \implies \left(X^{-\frac{1}{2}}YX^{-\frac{1}{2}}\right)^{\frac{1}{2}} &\ge X^{-\frac{1}{2}}NX^{-\frac{1}{2}}\\
    \implies X^{\frac{1}{2}}\left(X^{-\frac{1}{2}}YX^{-\frac{1}{2}}\right)^{\frac{1}{2}}X^{\frac{1}{2}} &\ge N,
\end{align}
where the second and the final inequalities follow from the fact that $A\ge 0 \implies BAB \ge 0$ for any positive semidefinite operator $B$, and the penultimate inequality follows from the fact that square root is operator monotone. As such, the positive semidefiniteness of the matrix $\begin{pmatrix}
    X & N\\
    N & Y
    \end{pmatrix}$ ensures that $N \le G_{\frac{1}{2}}(X,Y)$.
    
    Also note that $N = G_{\frac{1}{2}}(X,Y)$ satisfies $\begin{pmatrix}
    X & N\\
    N & Y
    \end{pmatrix} \ge 0$. This can be verified as follows:
    \begin{align}
        G_{\frac{1}{2}}(X,Y) X^{-1} G_{\frac{1}{2}}(X,Y) &= X^{\frac{1}{2}}\!\left(X^{-\frac{1}{2}}YX^{-\frac{1}{2}}\right)^{\frac{1}{2}}\!\left(X^{-\frac{1}{2}}YX^{-\frac{1}{2}}\right)^{\frac{1}{2}}X^{\frac{1}{2}}\\
        &= X^{\frac{1}{2}}X^{-\frac{1}{2}}YX^{-\frac{1}{2}}X^{\frac{1}{2}}\\
        &\le Y,\label{eq:geo_mean_le_Y}
    \end{align}
    where the inequality follows from the fact that $\operatorname{supp}(X) \subseteq \operatorname{supp}(Y)$. Also, $(I-XX^{-1})G_{\frac{1}{2}}(X,Y) = 0$, which, along with~\eqref{eq:geo_mean_le_Y}, is sufficient to ensure that  $\begin{pmatrix}
    X & G_{\frac{1}{2}}(X,Y)\\
    G_{\frac{1}{2}}(X,Y) & Y
    \end{pmatrix} \ge 0$ from the Schur complement lemma.
    
Recalling the identity in~\eqref{eq:geo_wt_multiplicative}, we can write
\begin{equation}
    G_{2^{-\ell}}(X,Y) = G_{\frac{1}{2}}\!\left(X,G_{2^{1-\ell}}(X,Y)\right).
\end{equation}
Now consider the following pair of inequalities:
\begin{equation}
    \begin{pmatrix}
    X & N_1\\
    N_1 & Y
    \end{pmatrix} \ge 0, \qquad \begin{pmatrix}
    X & N_2\\
    N_2 & N_1
    \end{pmatrix} \ge 0.
\end{equation}
The first inequality is satisfied if and only if $N_1 \le G_{\frac{1}{2}}(X,Y)$, and the second inequality is satisfied if and only if 
\begin{align}
    N_2 &\le G_{\frac{1}{2}}\!\left(X,N_1\right)\\
    &\le G_{\frac{1}{2}}\!\left(X,G_{\frac{1}{2}}(X,Y)\right)\\
    &= G_{\frac{1}{4}}(X,Y),
\end{align}
where the second inequality follows from~\eqref{eq:wt_geo_mean_monotonic} and the equality follows from~\eqref{eq:geo_wt_multiplicative}. From a recursive argument,
\begin{equation}\label{eq:geo_2^l_constraints}
    \begin{pmatrix}
    X & N_1\\
    N_1 & Y
    \end{pmatrix} \ge 0, \qquad \begin{pmatrix}
    X & N_{i+1}\\
    N_{i+1} & N_i
    \end{pmatrix} \ge 0 \qquad \forall i\in \{1,2,\ldots,\ell-1\}
\end{equation}
hold if and only if
\begin{equation}
    N_i \le G_{2^{-i}}(X,Y) \qquad \forall i\in \{1,2,\ldots,\ell\}.
\end{equation}

Now consider the following matrix inequality 
\begin{equation}
    \begin{pmatrix}
    L & X\\
    X & Y
    \end{pmatrix} \ge 0 
\end{equation}
which is satisfied by a Hermitian operator $L$ if and only if 
\begin{equation}
    L \ge XY^{-1}X = G_{-1}(X,Y).
\end{equation}
Adding this condition to the list of conditions given in~\eqref{eq:geo_2^l_constraints}, we can state that 
\begin{equation}\label{eq:geo_mean_2^l_constraints}
    \begin{pmatrix}
    L & X\\
    X & N_{\ell}
    \end{pmatrix} \ge 0, \qquad \begin{pmatrix}
    X & N_1\\
    N_1 & Y
    \end{pmatrix} \ge 0, \qquad \begin{pmatrix}
    X & N_{i+1}\\
    N_{i+1} & N_i
    \end{pmatrix} \ge 0 \qquad \forall i\in \{1,2,\ldots,\ell-1\}
\end{equation}
are satisfied if and only if
\begin{equation}
    N_i \le G_{2^{-i}}(X,Y) \qquad \forall i\in \{1,2,\ldots,\ell\}
\end{equation}
and 
\begin{align}
    L &\ge G_{-1}\!\left(X,N_{\ell}\right)\\
    &\ge G_{-1}\!\left(X,G_{2^{-\ell}}(X,Y)\right)\\
    &= G_{-2^{-\ell}}(X,Y),\label{eq:L_ge_geo_mean}
\end{align}
where the second inequality follows from the anti-monotonicity of $G_t(X,Y)$ when $t\in [-1,0)$.

We are interested in an SDP to compute $\left\Vert \operatorname{Tr}_B\!\left[G_{-2^{-\ell}}\!\left(J^{\mathcal{N}}_{RB},J^{\mathcal{M}}_{RB}\right)\right]\right\Vert_{\infty}$ (refer to~\eqref{eq:ch_geo_entropy_geo_mean_op_norm_form} for details). For a positive semidefinite operator $Z$ and a real number $y$, $y\openone \ge Z$ if and only if $y\ge \Vert Z\Vert_{\infty}$. Therefore, the following set of conditions:
\begin{equation}\label{eq:wt_geo_mean_op_norm_const}
    \begin{pmatrix}
    L_{RB} & J^{\mathcal{N}}_{RB}\\
    J^{\mathcal{N}}_{RB} & N_{\ell}
    \end{pmatrix} \ge 0, \qquad \begin{pmatrix}
    J^{\mathcal{N}}_{RB} & N_1\\
    N_1 & J^{\mathcal{M}}_{RB}
    \end{pmatrix} \ge 0, \qquad \left\{\begin{pmatrix}
    X & N_{i+1}\\
    N_{i+1} & N_i
    \end{pmatrix} \ge 0\right\}_{i=1}^{\ell}, \qquad y\openone_R \ge \operatorname{Tr}_B\!\left[L_{RB}\right]
\end{equation}
are satisfied if and only if
\begin{align}
    y &\ge \left\Vert\operatorname{Tr}_B\!\left[L_{RB}\right]\right\Vert_{\infty}\\
    &\ge  \left\Vert\operatorname{Tr}_B\!\left[G_{-2^{-\ell}}\!\left(J^{\mathcal{N}}_{RB},J^{\mathcal{M}}_{RB}\right)\right]\right\Vert_{\infty},
\end{align}
where the last inequality follows from~\eqref{eq:L_ge_geo_mean}. Therefore, minimizing over every real number $y$ and Hermitian operators $\left(N_1,N_2,\ldots,N_{\ell}\right)$ subject to the constraints in~\eqref{eq:wt_geo_mean_op_norm_const} is an SDP, the optimal value of which is equal to $\left\Vert \operatorname{Tr}_B\!\left[G_{-2^{-\ell}}\!\left(J^{\mathcal{N}}_{RB},J^{\mathcal{M}}_{RB}\right)\right]\right\Vert_{\infty}$. Using the aforementioned SDP, along with the closed form expression of the geometric R\'enyi channel divergence from~\eqref{eq:ch_geo_entropy_geo_mean_op_norm_form}, yields the program in~\eqref{eq:lemma3fang}.

\subsection{Primal SDP}
\label{sec:grd_sdp_primal}

Since $\widehat D_{\alpha,H,E}$ upper bounds all other channel divergences, we would like to calculate the smallest value of this quantity for the smallest possible $\alpha$ to obtain the tightest upper bound.  We define an operator $C$ as follows
\begin{align}\label{eq:ch_ec_grd_w_C}
\widehat D_{\alpha,H,E}\!\left(\curlN_{A\to B}\middle\|\curlM_{A\to B}\right) &= \frac{1}{\alpha -1} \ln  \sup_{\substack{\rho_R\ge 0, \operatorname{Tr}\left[\rho_R\right] = 1,\\ \operatorname{Tr}\left[H\rho\right] \le E}}
\tr[\rho_R \underbrace{\tr_B[G_{1-\alpha}(J^{\mathcal{N}}_{RB},J^{\mathcal{M}}_{RB})]}_{\equiv C}] 
\end{align}

To simplify the notation, we omit the subscripts $RB$ and $R$ below for clarity.

For the purpose of binary hypothesis testing where the channels $\curlN$ and $\curlM$ are known, we can compute 
$G_{1-\alpha}(J^\curlN_{RB},J^\curlM_{RB}$) directly, and we need only to optimize over the input state $\rho$. Therefore the energy-constrained geometric R\'enyi channel divergence can be efficiently computed as follows:
\begin{align}\label{eq:ch_grd_w_C_opt_form}
 \widehat{D}_{\alpha,H,E}\!\left(\mathcal{N}_{A\to B}\middle\Vert\mathcal{M}_{A\to B}\right) = \frac{1}{\alpha -1} \ln   
 \sup_{\rho \geq 0}
 \left \{\tr[\rho C]: \tr[\rho]=1,  ~\quad \tr[H\rho]\leq E\right\},
\end{align}
where $C$ is defined in~\eqref{eq:ch_ec_grd_w_C} and $H= \sum_{n=0}^\text{cutoff} n \ket{n}\!\bra{n}$ is the Hamiltonian.
Because we work with a truncated Hilbert space, both $C$ and $H$ are finite-dimensional matrices, and $\rho_R$ is optimized over finite-dimensional density operators.

For other scenarios where one of the channels may need to be optimized (e.g. in quantum reading or composite hypothesis testing, where $G_{1-\alpha} $ is part of the optimization), the program is more complex, and we include it here for completeness.

Introducing Lagrange multipliers, the optimization in~\eqref{eq:ch_grd_w_C_opt_form} can be written as follows:
\begin{align}
 \widehat{D}_{\alpha,H,E}\!\left(\mathcal{N}_{A\to B}\middle\Vert\mathcal{M}_{A\to B}\right) &= \frac{1}{\alpha -1} \ln   
 \left \{ \sup_{\rho\geq 0}\tr[\rho C] +
\inf_{\lambda \in \mathbb{R}}\!\left(\lambda\!\left(1-\tr[\rho]\right)\right)
+\inf_{\mu \geq 0}\left(\mu\!\left(E-\tr[H \rho]\right)\right)\right\}\\ 
&=
\frac{1}{\alpha -1} \ln    \sup_{\rho\geq 0}
\inf_{\substack{\lambda \in \mathbb{R},\\\mu \geq 0}}
 \left \{ \lambda+\mu E+
\tr\!\left[\rho (C-\lambda \openone - \mu H)\right]\right\}\\
&\le \frac{1}{\alpha -1} \ln    \inf_{\substack{\lambda \in \mathbb{R},\\\mu \geq 0}}\sup_{\rho\geq 0}
\left \{ \lambda+\mu E+
\tr\!\left[\rho (C-\lambda \openone - \mu H)\right]\right\}\\
&= \frac{1}{\alpha -1} \ln   
\inf_{\substack{\lambda \in \mathbb{R},\\\mu \geq 0}}
 \left \{
 \lambda+\mu E :~
 C\leq \lambda \openone + \mu H
\right\}, 
\end{align}
where the inequality follows from the max-min inequality. Setting $\alpha = 1+2^{-\ell}$ and using the semidefinite representation of $G_{-2^{-\ell}}\!\left(J^{\mathcal{N}}_{RB},J^{\mathcal{M}}_{RB}\right)$ from~\eqref{eq:geo_mean_2^l_constraints}, we arrive at the following semidefinite program for the energy-constrained geometric R\'enyi divergence of channels:
\begin{align}
\label{eq:final-SDP-GRD-composite}
\widehat{D}_{1+2^{-\ell}, H,E}\!\left(\mathcal{N}_{A\to B}\Vert\mathcal{M}_{A\to B}\right) =  2^\ell \ln   
\inf_{\substack{\lambda \in \mathbb{R}, ~\mu \geq 0 \\ L_{RB},N_0,\ldots, N_\ell\in \operatorname{Herm}}}
 \left \{
 \lambda +\mu E 
\right\} 
\end{align}
subject to the constraints
\begin{align}
            \left\{ C = \tr_B [L_{RB}] \leq \lambda \openone + \mu H, \quad
\left[
\begin{array}{cc}
L_{RB} & J^\curlN \\
 J^\curlN & N_\ell \\
\end{array}
\right] 
\geq 0, 
 \quad
\left\{ \left[
\begin{array}{cc}
 J^\curlN & N_i \\
 N_i & N_{i-1} \\
\end{array}
\right] \geq 0\right\}_{i=1}^\ell,
\qquad
 N_0 = J^\curlM \right\} .
\end{align}

\subsection{Dual SDP}

In this section, we derive the dual of the semidefinite program given in~\eqref{eq:final-SDP-GRD-composite}.

Recalling the semidefinite representation of the operator geometric mean from~\eqref{eq:geo_mean_2^l_constraints}, we can write the energy-constrained geometric R\'enyi divergence of channels as the following optimization:
\begin{equation}\label{eq:GRD_ch_sup_inf}
    \widehat{D}_{1+2^{-\ell},H,E}\!\left(\mathcal{N}_{A\to B}\middle\Vert\mathcal{M}_{A\to B}\right) = 2^{\ell}\ln\sup_{\substack{\rho_R\ge 0, \operatorname{Tr}[\rho] = 1,\\ \operatorname{Tr}[H\rho] \le E}}\inf_{L_{RB},N_0,\ldots, N_\ell\in \operatorname{Herm}}\left\{
    \begin{array}{c}
        \operatorname{Tr}\!\left[\left(\rho_R\otimes \openone_B\right)L_{RB}\right]:\\
        
\left[
\begin{array}{cc}
L_{RB} & J^\curlN \\
 J^\curlN & N_\ell \\
\end{array}
\right] 
\geq 0, \quad N_0 = J^{\mathcal{M}},\\ 
\left[
\begin{array}{cc}
 J^\curlN & N_i \\
 N_i & N_{i-1} \\
\end{array}
\right] \geq 0 \qquad \forall i = \{1,2,\ldots,\ell\}
    \end{array}\right\}.
\end{equation}
Let us focus on the inner infimum and find the dual of this optimization. Introducing Lagrange multipliers, we can write the inner infimum of~\eqref{eq:GRD_ch_sup_inf} as follows:
\begin{multline}\label{eq:GRD_ch_inf_sup}
    \inf_{L_{RB},N_0,\ldots, N_\ell\in \operatorname{Herm}}\sup\!\left\{\operatorname{Tr}\!\left[\left(\rho_R\otimes \openone_B\right)L_{RB}\right] - \operatorname{Tr}\!\left[\left[\begin{array}{cc}
    Y_{\ell} & W^{\dagger}_{\ell}\\
    W_{\ell} & Z_{\ell}
    \end{array}\right]\!\left[\begin{array}{cc}
    L & J^{\mathcal{N}}\\
    J^{\mathcal{N}} & N_{\ell}
    \end{array}\right]\right]\right.\\ \left.- \sum_{i=1}^{\ell}\operatorname{Tr}\!\left[\left[\begin{array}{cc}
    Y_{i-1} & W^{\dagger}_{i-1}\\
    W_{i-1} & Z_{i-1}
    \end{array}\right]\!\left[\begin{array}{cc}
    J^{\mathcal{N}} & N_i\\
    N_i & N_{i-1}
    \end{array}\right]\right] - \operatorname{Tr}\!\left[X\!\left(N_0 - J^{\mathcal{M}}\right)\right]: X \in \operatorname{Herm},\left[\begin{array}{cc}
    Y_{i} & W^{\dagger}_{i}\\
    W_{i} & Z_{i}
    \end{array}\right] \ge 0 \quad \forall i\in \{0,1,\ldots,\ell\} \right\}.
\end{multline}
Upon expanding and rearranging the expression in~\eqref{eq:GRD_ch_inf_sup}, we arrive at the following:
\begin{align}
&\inf_{L_{RB},N_0,\ldots, N_\ell\in \operatorname{Herm}}\sup\!\left\{\operatorname{Tr}\!\left[\left(\rho_R\otimes I_B - Y_{\ell}\right)L_{RB}\right] - \operatorname{Tr}\!\left[J^{\mathcal{N}}\!\left(W_{\ell} + W^{\dagger}_{\ell}+\sum_{i=0}^{\ell - 1}Y_i\right)\right] + \operatorname{Tr}\!\left[XJ^{\mathcal{M}}\right] - \operatorname{Tr}\!\left[N_0\!\left(X+Z_0\right)\right]\right.\notag\\
&\qquad\qquad\qquad\left. + \sum_{j=1}^{\ell}\operatorname{Tr}\!\left[N_i\!\left(Z_i + W_{i-1} + W^{\dagger}_{i-1}\right)\right]: X\in \operatorname{Herm},\left[\begin{array}{cc}
    Y_i & W_i^{\dagger} \\
    W_i & Z_i
\end{array}\right]\ge 0 \quad \forall i\in \{0,1,\ldots, \ell\}\right\}\\
&\ge \sup_{\substack{X\in \operatorname{Herm},\\ \left\{\left[\begin{array}{cc}
    Y_i & W_i^{\dagger} \\
    W_i & Z_i
\end{array}\right]\ge 0\right\}_{i=0}^{\ell}}}\inf_{L_{RB},N_0,\ldots, N_\ell\in \operatorname{Herm}}\!\left\{\operatorname{Tr}\!\left[\left(\rho_R\otimes I_B - Y_{\ell}\right)L_{RB}\right] - \operatorname{Tr}\!\left[J^{\mathcal{N}}\!\left(W_{\ell} + W^{\dagger}_{\ell}+\sum_{i=0}^{\ell - 1}Y_i\right)\right]\right.\notag\\
&\qquad\qquad\qquad\qquad\qquad \left. + \operatorname{Tr}\!\left[XJ^{\mathcal{M}}\right]+ \sum_{j=1}^{\ell}\operatorname{Tr}\!\left[N_i\!\left(Z_i + W_{i-1} + W^{\dagger}_{i-1}\right)\right] - \operatorname{Tr}\!\left[N_0\!\left(X+Z_0\right)\right]\right\},\label{eq:GRD_ch_sup_inf_inner}
\end{align}
where we have used the max-min inequality to swap the order of the infimum and the supremum. The expression in~\eqref{eq:GRD_ch_sup_inf_inner} is finite if and only if $\rho_R\otimes I_B = Y_{\ell}$, $Z_i + W_{i-1} + W^{\dagger}_{i-1} = 0$, and $X = -Z_0$. Therefore, the optimization in~\eqref{eq:GRD_ch_sup_inf_inner} is equivalent to the following semidefinite program:
\begin{equation}
    \sup_{\left\{\left[\begin{array}{cc}
    Y_i & W_i^{\dagger} \\
    W_i & Z_i
\end{array}\right]\ge 0\right\}_{i=0}^{\ell}}\left\{ 
    \begin{array}{c}
    -\operatorname{Tr}\!\left[Z_0J^{\mathcal{M}}\right] - \operatorname{Tr}\!\left[J^{\mathcal{N}}\!\left(W_{\ell} + W^{\dagger}_{\ell}+\sum_{i=0}^{\ell - 1}Y_i\right)\right]:\\
    Y_{\ell} = \rho_R\otimes I_B,\\
    Z_i + W_{i-1} + W^{\dagger}_{i-1} = 0 \quad \forall i\in \{1,2,\ldots,\ell\}
    \end{array}\right\}.
\end{equation}
Since the above optimization yields a value less than or equal to the value of the inner infimum in~\eqref{eq:GRD_ch_sup_inf}, we can write the following semidefinite programming lower bound on the geometric R\'enyi divergence of channels:
\begin{equation}
    \widehat{D}_{1+2^{-\ell},H,E}\!\left(\mathcal{N}_{A\to B}\middle\Vert\mathcal{M}_{A\to B}\right) \ge 2^{\ell}\ln\sup_{\substack{\rho_R\ge 0,\\ \left\{\left[\begin{array}{cc}
    Y_i & W_i^{\dagger} \\
    W_i & Z_i
\end{array}\right]\ge 0\right\}_{i=0}^{\ell}}}\left\{ 
    \begin{array}{c}
    -\operatorname{Tr}\!\left[Z_0J^{\mathcal{M}}\right] - \operatorname{Tr}\!\left[J^{\mathcal{N}}\!\left(W_{\ell} + W^{\dagger}_{\ell}+\sum_{i=0}^{\ell - 1}Y_i\right)\right]:\\
    Y_{\ell} = \rho_R\otimes I_B, \operatorname{Tr}[\rho] = 1, \operatorname{Tr}[H\rho] \le E,\\
    Z_i + W_{i-1} + W^{\dagger}_{i-1} = 0 \quad \forall i\in \{1,2,\ldots,\ell\}
    \end{array}\right\}.
\end{equation}
We can replace $W_i$ with $-W_i$ for every $i\in \{0,1,\ldots,\ell\}$ in the above expression to arrive at the following expression:
\begin{equation}\label{eq:GRD_ch_dual}
    \widehat{D}_{1+2^{-\ell},H,E}\!\left(\mathcal{N}_{A\to B}\middle\Vert\mathcal{M}_{A\to B}\right) \ge 2^{\ell}\ln\sup_{\substack{\rho_R\ge 0,\\ \left\{\left[\begin{array}{cc}
    Y_i & W_i^{\dagger} \\
    W_i & Z_i
\end{array}\right]\ge 0\right\}_{i=0}^{\ell}}}\left\{ 
    \begin{array}{c}
    \operatorname{Tr}\!\left[J^{\mathcal{N}}\!\left(W_{\ell} + W^{\dagger}_{\ell}-\sum_{i=0}^{\ell - 1}Y_i\right)\right]-\operatorname{Tr}\!\left[Z_0J^{\mathcal{M}}\right] :\\
    Y_{\ell} = \rho_R\otimes I_B, \operatorname{Tr}[\rho] = 1, \operatorname{Tr}[H\rho] \le E,\\
    Z_i = W_{i-1} + W^{\dagger}_{i-1} \quad \forall i\in \{1,2,\ldots,\ell\}
    \end{array}\right\}.
\end{equation}


\section{Hilbert space truncation for the loss-dephasing channel}

\label{app:trunc-plot}

Fig.~\ref{fig:dim_rho_r_app} indicates, for loss-dephasing bosonic channels, that the truncation error decreases as the truncation parameter increases, similar to what was plotted in Fig.~\ref{fig:dim_rho_r} in the main text.

\begin{figure}[h!]
  \centering
  \includegraphics[trim= {0 0 0 0cm},clip, width=0.5\linewidth]{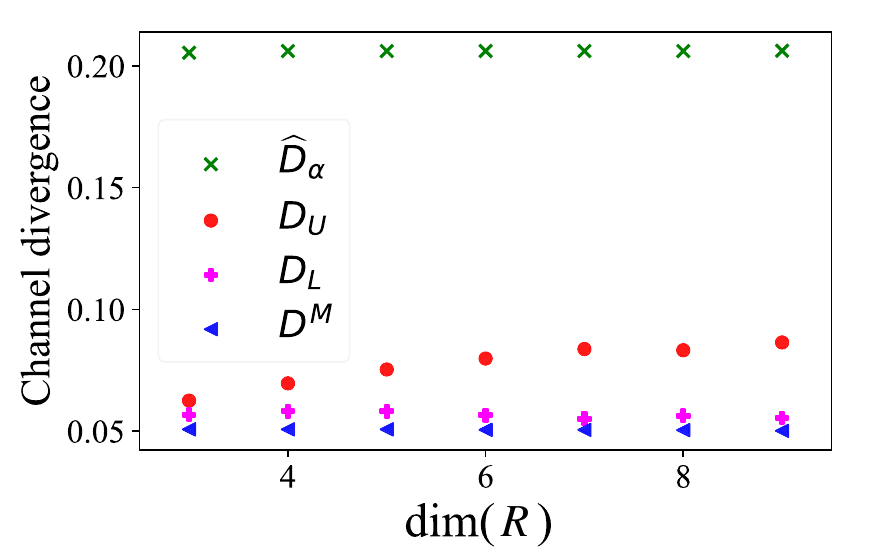}
  \caption{\label{fig:dim_rho_r_app} Channel divergences for two loss-dephasing channels as a function of the Hilbert space truncation, for $E=1$; the transmissivity and dephasing parameters are as follows: $\eta_1 = 0.95, \eta_2 = 0.85$, $\gamma_1 =\gamma_2= 0.01$}
\end{figure}

\section{Truncation error bound}

\label{app:trunc-err-bnd}

In this appendix, we prove that the error induced by truncation, when calculating the energy-constrained Belavkin--Staszewski divergence of bosonic dephasing channels, is no larger than $\frac{2E}{N+1}D(p\|q)$, where $E$ is the energy constraint, $N$ is the truncation parameter, and $D(p\|q)$ is the relative entropy of the probability densities underlying the bosonic dephasing channels. It should be noted that, if the input state is truncated in the Fock basis, then the output state is also, because bosonic dephasing channels do not increase photon number (they only randomize phases). 

\begin{proposition}
Let $\mathcal{N}$ and $\mathcal{M}$ be bosonic dephasing channels
with underlying probability densities $p$ and $q$, respectively.
Let $\hat{n}$ be the photon number operator, $E\geq0$ an energy constraint, and $N\in\mathbb{N}$
a truncation parameter. Then
\begin{equation}
\widehat{D}_{\hat{n},E,N}(\mathcal{N}\|\mathcal{M}) \leq  \widehat{D}_{\hat{n},E}(\mathcal{N}\|\mathcal{M})\leq\widehat{D}_{\hat{n},E,N}(\mathcal{N}\|\mathcal{M})+\frac{2E}{N+1}D(p\|q),\label{eq:trunc-bnd}
\end{equation}
where $\widehat{D}_{\hat{n},E,N}$ denotes the truncated, energy-constrained Belavkin--Staszewski divergence:
\begin{equation}
\widehat{D}_{\hat{n},E,N}(\mathcal{N}\|\mathcal{M})\coloneqq
\sup_{\substack{\psi_{RA}:\Tr[H_{A}\psi_{A}]\leq E,\\
\psi_{A}\in\mathcal{D}_{N}
}
}\widehat{D}((\operatorname{id}_R\otimes\mathcal{N})(\psi_{RA})\|(\operatorname{id}_R\otimes\mathcal{M})(\psi_{RA})),
\end{equation}
$\mathcal{D}_{N}$ denotes the set of density operators with support
in $\operatorname{span}\{|0\rangle,|1\rangle,\ldots,|N\rangle\}$, and $D(p\|q)$
denotes the relative entropy of $p$ and $q$, defined as
\begin{equation}
D(p\|q)\coloneqq\int_{-\pi}^{\pi}d\phi\,p(\phi)\ln\!\left(\frac{p(\phi)}{q(\phi)}\right).
\end{equation}
\end{proposition}

\begin{proof}
The first inequality in \eqref{eq:trunc-bnd} follows because the set $\mathcal{D}_N$ is contained in the set of all density operators.

To see the other inequality in \eqref{eq:trunc-bnd}, 
consider that
\begin{align}
  \widehat{D}_{\hat{n},E,N}(\mathcal{N}\|\mathcal{M})
 & =\sup_{\substack{\rho_{R}:\Tr[\hat{n}_{R}\rho_{R}]\leq E,\\
\rho_{R}\in\mathcal{D}_{N}
}
}\Tr\!\left[\left(\rho_{R}\otimes \openone_{B}\right)D_{\mathrm{op}}(J_{RB}^{\mathcal{N}}\|J_{RB}^{\mathcal{M}})\right]\\
 & =\sup_{\substack{\rho_{R}:\Tr[\hat{n}_{R}\rho_{R}]\leq E,\\
\rho_{R}\in\mathcal{D}_{N}
}
}\Tr\!\left[\rho_{R}\Tr_{B}\!\left[D_{\mathrm{op}}(J_{RB}^{\mathcal{N}}\|J_{RB}^{\mathcal{M}})\right]\right].
\end{align}
by the same reasoning that led to \eqref{eq:BS-rel-ent-D-op}.

Let $\sigma$ be an arbitrary number-diagonal state satisfying $\Tr[\hat{n}\sigma]\leq E$ (indeed, as argued previously, number-diagonal states are optimal for this optimization problem, due to the phase-covariance symmetric of bosonic dephasing channels).
Let $\sigma^{N}\coloneqq\frac{\Pi_{N}\sigma\Pi_{N}}{\Tr[\Pi_{N}\sigma]}$,
where $\Pi_{N}$ is the projection onto $\operatorname{span}\{|0\rangle,|1\rangle,\ldots,|N\rangle\}$.
Consider that
\begin{align}
\frac{1}{2}\left\Vert \sigma-\sigma_{N}\right\Vert _{1} & =\frac{1}{2}\left\Vert \sigma-\frac{\Pi_{N}\sigma\Pi_{N}}{\Tr[\Pi_{N}\sigma]}\right\Vert _{1}\\
 & =\frac{1}{2}\left\Vert \sigma-\Pi_{N}\sigma\Pi_{N}+\Pi_{N}\sigma\Pi_{N}-\frac{\Pi_{N}\sigma\Pi_{N}}{\Tr[\Pi_{N}\sigma]}\right\Vert _{1}\\
 & \leq\frac{1}{2}\left\Vert \sigma-\Pi_{N}\sigma\Pi_{N}\right\Vert _{1}+\frac{1}{2}\left\Vert \Pi_{N}\sigma\Pi_{N}-\frac{\Pi_{N}\sigma\Pi_{N}}{\Tr[\Pi_{N}\sigma]}\right\Vert _{1}\\
 & =\frac{1}{2}\Tr[(I-\Pi_{N})\sigma]+\frac{1}{2}\left|1-\frac{1}{\Tr[\Pi_{N}\sigma]}\right|\left\Vert \Pi_{N}\sigma\Pi_{N}\right\Vert _{1}\\
 & =\frac{1}{2}\Tr[(I-\Pi_{N})\sigma]+\frac{1}{2}\left|1-\frac{1}{\Tr[\Pi_{N}\sigma]}\right|\Tr[\Pi_{N}\sigma]\\
 & =\frac{1}{2}\Tr[(I-\Pi_{N})\sigma]+\frac{1}{2}\left|\Tr[\Pi_{N}\sigma]-1\right|\\
 & =\Tr[(I-\Pi_{N})\sigma]\\
 & \leq\frac{E}{N+1},
\end{align}
where the last inequality follows because $\hat{n}\geq\left(N+1\right)\left(I-\Pi_{N}\right)$,
which implies that $\Tr[\hat{n}\sigma]\geq\left(N+1\right)\Tr[\left(I-\Pi_{N}\right)\sigma]$.
An application of the H\"older inequality implies that
\begin{align}
  \Tr\!\left[\sigma_{R}\Tr_{B}\!\left[D_{\mathrm{op}}(J_{RB}^{\mathcal{N}}\|J_{RB}^{\mathcal{M}})\right]\right]
 & \leq\Tr\!\left[\sigma_{R}^{N}\Tr_{B}\!\left[D_{\mathrm{op}}(J_{RB}^{\mathcal{N}}\|J_{RB}^{\mathcal{M}})\right]\right]+\frac{2E}{N+1}\left\Vert \Tr_{B}\!\left[D_{\mathrm{op}}(J_{RB}^{\mathcal{N}}\|J_{RB}^{\mathcal{M}})\right]\right\Vert _{\infty}\\
 & =\Tr\!\left[\sigma_{R}^{N}\Tr_{B}\!\left[D_{\mathrm{op}}(J_{RB}^{\mathcal{N}}\|J_{RB}^{\mathcal{M}})\right]\right]+\frac{2E}{N+1}\widehat{D}(\mathcal{N}\|\mathcal{M})\\
 & \leq\widehat{D}_{\hat{n},E,N}(\mathcal{N}\|\mathcal{M})+\frac{2E}{N+1}D(p\|q).
\end{align}
The last inequality above follows because $\widehat{D}(\mathcal{N}\|\mathcal{M}) = D(p\|q)$ for all bosonic dephasing channels, as a consequence of the arguments of \cite{huang2024exact} (indeed the same arguments from Section~V therein apply to the Belavkin--Staszewski divergence). 
Since the bound above applies to an arbitrary number-diagonal state $\sigma$ satisfying
$\Tr[\hat{n}\sigma]\leq E$, we conclude \eqref{eq:trunc-bnd} after taking
the supremum over all such states.
\end{proof}

\medskip

While it would be ideal to have a similar kind of truncation error bound for pure-loss bosonic channels, an argument along the lines above is not applicable: indeed, the unconstrained Belavkin--Staszewski divergence of two distinct pure-loss channels (i.e., with transmissivity $\eta_1, \eta_2 \in [0,1]$ such that $\eta_1 \neq \eta_2$) is infinite, as a consequence of sending in coherent states of unbounded energy. A similar construction was used in \cite{winter2017} to argue that the diamond distance without energy constraints is equal to its maximum value of two.

\end{document}